\pdfoutput=1
\RequirePackage{ifpdf}
\ifpdf % We are running pdfTeX in pdf mode
\documentclass[pdftex]{sigma}
\else
\documentclass{sigma}
\fi

\usepackage{epic}
\usepackage{accents}

\numberwithin{equation}{section}

\newtheorem{Theorem}{Theorem}[section]
\newtheorem{Lemma}[Theorem]{Lemma}
\newtheorem{Proposition}[Theorem]{Proposition}
 { \theoremstyle{definition}
\newtheorem{Remark}[Theorem]{Remark} }

\def \h#1{\widehat{#1}}
\def \t#1{\widetilde{#1}}
\def \b#1{\overline{#1}}
\def \th#1{\widehat{\widetilde{#1}}}
\def \hb#1{{\overline{\widehat{#1}}}}
\def \bh#1{{\widehat{\overline{#1}}}}

\def \tb#1{\overline{\widetilde{#1}}}
\def \bt#1{\widetilde{\overline{#1}}}
\def \tt#1{\widetilde{\widetilde{#1}}}
\def \thb#1{\overline{\widehat{\widetilde{#1}}}}
\def \tthb#1{\overline{\widehat{\widetilde{\widetilde{#1}}}}}
\def \dh#1{\underaccent{\hat}{#1}}
\def \db#1{\underaccent{\bar}{#1}}
\def \dt#1{\underaccent{\tilde}{#1}}
\def \dth#1{\underaccent{\tilde}{\underaccent{\hat}{#1}}}
\def \dtb#1{\underaccent{\bar}{\underaccent{\tilde}{#1}}}
\def \dhb#1{\underaccent{\bar}{\underaccent{\hat}{#1}}}

\def \tth#1{\widehat{\widetilde{\widetilde{#1}}}}
\def \ttb#1{\bar{\widetilde{\widetilde{#1}}}}

\def \c#1{\accentset{\circ}{#1}}
\def \flo#1{{\left\lfloor{#1}\right\rfloor}}

\begin{document}

\allowdisplaybreaks

\newcommand{\arXivNumber}{1912.00713}

\renewcommand{\PaperNumber}{060}

\FirstPageHeading

\ShortArticleName{Multi-Component Extension of CAC Systems}

\ArticleName{Multi-Component Extension of CAC Systems}

\Author{Dan-Da ZHANG~$^\dag$, Peter H.~VAN DER KAMP~$^\ddag$ and Da-Jun ZHANG~$^\S$}

\AuthorNameForHeading{D.-D.~Zhang, P.H.~van der Kamp and D.-J.~Zhang}

\Address{$^\dag$~School of Mathematics and Statistics, Ningbo University, Ningbo 315211, P.R.~China}
\EmailD{\href{mailto:zhangdanda@nbu.edu.cn}{zhangdanda@nbu.edu.cn}}

\Address{$^\ddag$~Department of Mathematics and Statistics, La Trobe University, Victoria 3086, Australia}
\EmailD{\href{mailto:P.vanderKamp@LaTrobe.edu.au}{P.vanderKamp@LaTrobe.edu.au}}

\Address{$^\S$~Department of Mathematics, Shanghai University, Shanghai 200444, P.R.~China}
\EmailD{\href{mailto:djzhang@staff.shu.edu.cn}{djzhang@staff.shu.edu.cn}}

\ArticleDates{Received December 07, 2019, in final form June 14, 2020; Published online July 01, 2020}

\Abstract{In this paper an approach to generate multi-dimensionally consistent $N$-com\-ponent systems is proposed. The approach starts from scalar multi-dimensionally consistent quadrilateral systems and makes use of the cyclic group. The obtained $N$-component systems inherit integrable features such as B\"acklund transformations and Lax pairs, and exhibit interesting aspects, such as nonlocal reductions. Higher order single component lattice equations (on larger stencils) and multi-component discrete Painlev\'e equations can also be derived in the context, and the approach extends to $N$-component generalizations of higher dimensional lattice equations.}


\Classification{37K60}

\section{Introduction}\label{sec-1}
Integrability of nonlinear partial differential equations is of sovereign
importance in the study of soliton theory.
For discrete equations, especially quadrilateral equations, consistency
around the cube (CAC) \cite{BS,Nij02, NW01} provides an interpretation
of integrability. A classification of discrete integrable equations with scalar-valued fields was obtained by Adler--Bobenko--Suris (ABS) in~\cite{ABS03} and more general classes of scalar equations where classified in \cite{ABS09,RB11}. The CAC property is also applicable to higher order lattice equations by introducing suitable multi-component forms \cite{Hietarinta-JPA-DBSQ,TN-2005,Walker-PhD}. Further examples of multi-component integrable systems can be found in the literature \cite{KNPT19,KNPT20,Kels, PT}. However, general classification results have not been obtained yet.

For CAC equations, the equations on the side faces of the consistent cube can be interpreted as an auto-B\"acklund transformation (BT).
CAC also enables one to construct Lax pairs \cite{BHQK12,Nij02}, as well as to find soliton solutions \cite{HZ09,HZ10,HZ11}.
Whereas the classification in \cite{ABS03} requires the equations on the six faces of the cube to be the same, in \cite{Atk08} alternative auto-BTs were given for several ABS equations, giving rise to consistent systems where the equations on the side faces are different from the equation on the top and bottom of the cube. Moreover, (non-auto) BTs between distinct equations were also provided, corresponding to consistent systems with different equations on the top and the bottom faces of the cube.
Other classifications of CAC systems with asymmetrical properties, other relaxations, and 3D affine linear lattice equations with 4D consistency have been also considered \cite{ABS09,ABS-IMRN-2011,RB11,Hietarinta-JNMP-2019,LMP-3D-2008}.
We will refer to a system of equations which is consistent on a cube, as a cube system.

In Appendix~B of the PhD.~Thesis of J.~Atkinson \cite{Jat}, multi-component versions of CAC scalar systems were introduced under the name ``The trivial Toeplitz extension''. Atkinson trivially extends a scalar equation for a field $u$ to an $N$-component system for fields $u_1,\ldots,u_N$, and then
applies the transformation
\[
u_i(l,m)\rightarrow u_{i+l+m \ {\rm mod} \ N}(l,m).
\]
He \looseness=-1 remarks that the resulting system is trivial (indeed the inverse transformation decouples it), and that extensions of multi-dimensional consistent scalar equations are multi-dimensional consistent (and hence would emerge in classifications of multi-component discrete integrable systems).

In \cite{FZ14}, $(N=2)$-component ABS equations resulting from such an extension were investigated. Here the CAC property (with affine linearity, D4 symmetry and the tetrahedron property) of these coupled systems was established, and solutions provided.

\looseness=-1 In this paper, we consider more general (but still trivial) extensions of not only scalar equations, but also of cube systems. For scalar equations these extension take the (non-Toeplitz) form
\[
u_i(l,m)\rightarrow u_{i+al+bm\ {\rm mod} \ N}(l,m), \qquad \text{with}\quad a,b\in \mathbb{Z}.
\]
We will provide Lax pairs, BTs, solutions and reductions for multi-component extensions of CAC lattice equations and cube systems, as well as multi-component extensions of 3D lattice equations with 4D consistency.

\looseness=-1 The paper is organized as follows. In Section~\ref{sec-2} we construct multi-component extensions of two distinguished kinds of systems, CAC {\em cube systems} (Section~\ref{sec-2-1}) and CAC {\em lattice systems} (quadrilateral equations in Section~\ref{sec-2-2}, and higher dimensional equations in Section~\ref{sec-6}). CAC lattice systems can be consistently posed on a lattice, whereas CAC cube systems need to be accompanied by reflected cube systems and posed on lattices similar to so-called black and white lattices~\cite{ABS09, XP}. Examples include multi-component extensions of a Boll cube system~\cite{RB11} (Section~\ref{sec-2-1}), an equation from the ABS list \cite{ABS03}, and (auto and non-auto) B\"acklund transformations~\cite{Atk08} (Section~\ref{sec-2-4}). We show that multi-component extensions admit Lax-pairs in Section~\ref{slax}. Assuming D4 symmetry we count the number of different non-decoupled $N$-component extensions in Section~\ref{sec-2-3}. In Section~\ref{sec-4} we provide several kinds of reductions. Nonlocal reductions are given in Section~\ref{sec-4-1}, reductions to higher order scalar equations are given in Section~\ref{sec-4-2}, and a reduction to a multi-component Painlev\'e lattice equation is considered in Section~\ref{sec-4-3}. In Section~\ref{sec-5-1} some particular solutions are given. These can be constructed from $N$ solutions of the scalar equation. A solution for a nonlocal equation is provided in Section~\ref{sec-5-3}. In Section~\ref{sec-7} we summarize and discuss the results, and we point out that particular examples of $N$-component generalised systems have appeared in the literature in different contexts.

\section{Multi-component extension of CAC systems}\label{sec-2}
In this section we construct multi-component systems that are consistent around the cube, a.k.a.~CAC. We first focus on a single 3D cube with consistent face equations.

\subsection{Multi-component CAC cube systems}\label{sec-2-1}
We will be concerned with quadrilateral equations of the form
\begin{gather}\label{Q}
Q\big(u,\t u,\h u,\th u\big)=0,
\end{gather}
where we use $\t u$ and $\h u$ to denote shifts of $u$ in two different directions. Posing six such equations on a cube yields a general type of system of the form
\begin{subequations}\label{consist-1}
\begin{alignat}{3}
& Q\big(u,\t u,\h u, \th u\big)=0,\qquad && Q^*\big(\b u,\tb u,\hb u, \thb u\big)=0,& \\
& A\big(u,\t u,\b u, \tb u\big)=0,\qquad &&A^*\big(\h u,\th u,\hb u, \thb u\big)=0,&\\
& B\big(u,\h u,\b u,\hb u\big)=0,\qquad && B^*\big(\t u,\th u,\tb u,\thb u\big)=0.&
\end{alignat}
\end{subequations}
We assume the functions $Q$, $A$, $B$, $Q^*$, $A^*$, $B^*$ are affine linear with respect to each variable, and the symbols $\t{u}, \h{u}, \b{u}, \dots, \thb{u}$ represent the values of the field at the vertices of the cube, see Fig.~\ref{F:1}(a).
Each equation may depend on additional (edge) parameters but we omit these.

The system \eqref{consist-1} is called CAC if the three values for $\thb u$ calculated from the three starred equations coincide for arbitrary initial data $u$, $\t u$, $\h u$, $\b u$, i.e.,
\begin{gather}
\thb u=F\big(u,\t u,\h u,\b u\big).
\label{F}
\end{gather}

In order to get a multi-component extension of the system~\eqref{consist-1}, we consider the vertex symbols $u$, $\t u$, $\h u$, $\th u$, $\tb u$, $\hb u$, $\thb u$ to be $N\times N$ diagonal matrices, e.g.,
\begin{gather}\label{u-form}
 u=\operatorname{Diag}(u_1,u_2,\dots,u_N),\qquad \t u=\operatorname{Diag}(\t u_1,\t u_2,\dots,\t u_N).
\end{gather}
We introduce a cyclic group using the generator $\sigma=\sigma_N$, defined as the $N\times N$ matrix with elements given by
\begin{gather}\label{sigma}
 (\sigma_N)_{i,j} = \begin{cases}1, & i+1 \equiv j \ \text{mod} \ N,\\0, & \text{otherwise}.
 \end{cases}
\end{gather}
Thus, a cyclic transformation of $u$ (a permutation of the components on the diagonal) can be denoted by
\begin{gather}\label{T}
u \mapsto T_k u =\sigma^k u \sigma^{-k},\qquad k\in \mathbb{Z} \ ({\rm mod} \ N).
\end{gather}
Note that $\sigma^N=I_N$ which is the $N\times N$ identity matrix.

\begin{Lemma}\label{T0}
If the scalar cube system \eqref{consist-1} is CAC, the following multi-component cube system~
\begin{subequations}\label{consist-2}
\begin{alignat}{3}
&Q\big(u, T_{k_1} \t u, T_{k_2} \h u, T_{k_4}\th u \big)=0,\qquad && Q^*\big(T_{k_3} \b u, T_{k_5} \tb u, T_{k_6} \hb u, T_{k_7} \thb u \big)=0,&\\
&A\big(u, T_{k_1} \t u, T_{k_3} \b u, T_{k_5} \tb u \big)=0,\qquad && A^*\big(T_{k_2} \h u, T_{k_4}\th u, T_{k_6} \hb u, T_{k_7} \thb u \big)=0,& \label{consist-2b}\\
&B\big(u, T_{k_3} \b u, T_{k_2} \h u, T_{k_6} \hb u \big)=0,\qquad && B^*\big(T_{k_1} \t u, T_{k_5}\tb u, T_{k_4} \th u, T_{k_7} \thb u \big)=0& \label{consist-2c}
\end{alignat}
\end{subequations}
is CAC as well, where the variables $u$, $\t u$, $\h u$, $\th u$, $\tb u$, $\hb u$, $\thb u$ are diagonal matrices as in \eqref{u-form},
and the~$T_{k_i}$ are cyclic transformations as defined in~\eqref{T}, with $k_i\in \mathbb{Z}$ $({\rm mod}~N)$.
\end{Lemma}

\begin{proof}Note that the equations in the system \eqref{consist-1} are affine linear and the variables denote the values of fields at the vertices of
the cube. Replacing these variables by diagonal matrices they remain commutative. If the system~\eqref{consist-1} is CAC and $\thb u$ has a unique expression~\eqref{F}, it then follows that the system~\eqref{consist-2} is CAC in the sense that $\thb u$ has a unique expression
\begin{gather*}
T_{k_7}\thb u =F\big(u,T_{k_1} \t u,T_{k_2} \h u, T_{k_3} \b u\big)
%\label{FF}
\end{gather*}
in which the function $F$ is the same as in equation \eqref{F}, and the fields at the vertices are relabelled.
\end{proof}

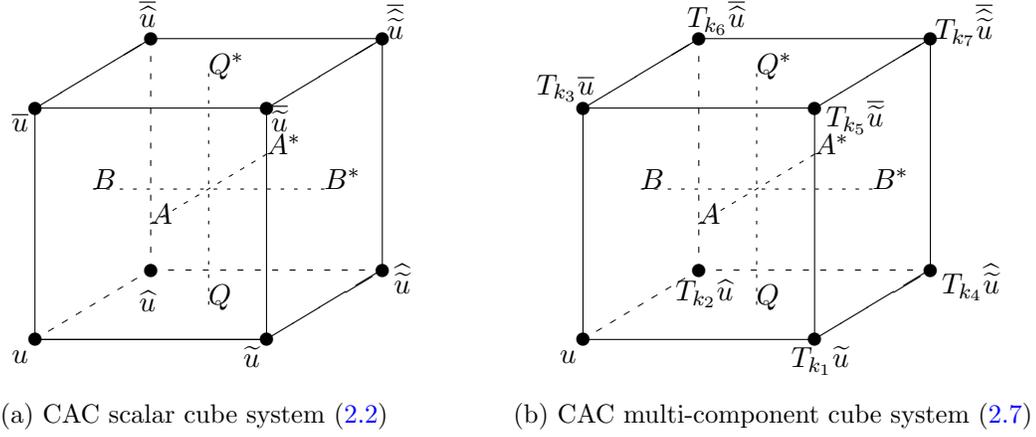
\begin{figure}[ht]
\setlength{\unitlength}{0.08em}
\hskip 2.0cm
\begin{picture}(200,170)(10,-20)
 \put(100, 0){\circle*{6}} \put(0 ,100){\circle*{6}}
 \put( 50, 30){\circle*{6}} \put(150,130){\circle*{6}}
 \put( 0, 0){\circle*{6}} \put(100,100){\circle*{6}}
 \put( 50,130){\circle*{6}} \put(150, 30){\circle*{6}}
 \put( 0, 0){\line(1,0){100}}
 \put( 0,100){\line(1,0){100}}
 \put(50,130){\line(1,0){100}}
 \put( 0, 0){\line(0,1){100}}
 \put(100, 0){\line(0,1){100}}
 \put(150,30){\line(0,1){100}}
 \put( 0,100){\line(5,3){50}}
 \put(100,100){\line(5,3){50}}
 \put(100, 0){\line(5,3){50}}
 \put( 0, 0){\line(1,0){100}}
 \dashline{3}(50,30)(0,0)
 \dashline{3}(50,30)(150,30)
\dashline{3}(50,30)(50,130)
 \dashline{2}(25, 65)(125, 65)
 \dashline{2}(75, 15)(75, 115)
 \dashline{2}(50, 50)(100, 80)
 \put(25, 65){$B$}
 \put(125, 65){$B^*$}
 \put(75, 15){$Q$}
 \put(75, 115){$Q^*$}
 \put(50, 50){$A$}
 \put(100, 80){$A^*$}
 \put(-10,-10){$u$}
 \put(90,-12){$\t u $}
 \put(45,12){$\h u$}
 \put(-10,90){$\b u$}
 \put(155,20){$\th u$}
 \put(45,135){$\hb u$}
 \put(102,90){$\tb u$}
 \put(152,130){$\thb u$}
 \put(-15,-36){\small (a) CAC scalar cube system \eqref{consist-1}}
\end{picture}
\hskip 1.0cm
\begin{picture}(200,170)(10,-20)
 \put(100, 0){\circle*{6}} \put(0 ,100){\circle*{6}}
 \put( 50, 30){\circle*{6}} \put(150,130){\circle*{6}}
 \put( 0, 0){\circle*{6}} \put(100,100){\circle*{6}}
 \put( 50,130){\circle*{6}} \put(150, 30){\circle*{6}}
 \put( 0, 0){\line(1,0){100}}
 \put( 0,100){\line(1,0){100}}
 \put(50,130){\line(1,0){100}}
 \put( 0, 0){\line(0,1){100}}
 \put(100, 0){\line(0,1){100}}
 \put(150,30){\line(0,1){100}}
 \put( 0,100){\line(5,3){50}}
 \put(100,100){\line(5,3){50}}
 \put(100, 0){\line(5,3){50}}
 \dashline{3}(50,30)(0,0)
 \dashline{3}(50,30)(150,30)
\dashline{3}(50,30)(50,130)
 \dashline{2}(25, 65)(125, 65)
 \dashline{2}(75, 15)(75, 115)
 \dashline{2}(50, 50)(100, 80)
 \put(25, 65){$B$}
 \put(125, 65){$B^*$}
 \put(75, 15){$Q$}
 \put(75, 115){$Q^*$}
 \put(50, 50){$A$}
 \put(100, 80){$A^*$}
 \put(-10,-10){$u$}
 \put(90,-12){$T_{k_1} \t u $}
 \put(40,17){$T_{k_2} \h u $}
 \put(-20,106){$T_{k_3} \b u $}
 \put(155,20){$T_{k_4} \th u $}
 \put(45,135){$T_{k_6} \hb u $}
 \put(105,92){$T_{k_5} \tb u $}
 \put(152,130){$T_{k_7} \thb u $}
 \put(-30,-36){\small (b) CAC multi-component cube system \eqref{consist-2}}
\end{picture}
\vspace{5mm}

\caption{Single cubes with equations defined on their faces.}\label{F:1}
\end{figure}

In Fig.~\ref{F:1} the cube on the right has equation $A\big(u, T_{k_1} \t u, T_{k_3} \b u, T_{k_5} \tb u \big)=0$ on its front face, which can be thought of in two distinct ways: as a relabeling of the variable names (which we do in the proof), or, as introducing coupling between different components of the fields at the vertices (which yield multi-component coupled systems of equations).

As an example we consider a cube system of Boll
\cite[equations~(3.31), (3.32)]{RB11}. With $N=2$, denoting the field
components by $u$, $v$ (instead of $u_1$, $u_2$), and taking
$k_i=\frac{1}{2}(1-(-1)^i)$ we find the following 2-component cube system
(written as vector system instead of as a matrix system):
\begin{gather*}
Q=\begin{pmatrix}
\h{u} \th{u} \delta_{{1}}+\h{u} \t{v} \delta_{{2}}+u\h{u}+\t{v} \th{u}\\
\t{u} \h{v} \delta_{{2}}+\h{v} \th{v} \delta_{{1}}+\t{u} \th{v}+v\h{v}
\end{pmatrix}=\begin{pmatrix}0\\0\end{pmatrix},\\
A=\begin{pmatrix}
\alpha \b{v} \tb{v} \delta_{{1}}+\b{v} \t{v}
\delta_{{2}}+u\b{v}+\t{v} \tb{v}\\
\alpha \b{u} \tb{u} \delta_{{1}}+\b{u} \t{u} \delta_{{2}}+v\b{u}+\t{u} \tb{u}
\end{pmatrix}=\begin{pmatrix}0\\0\end{pmatrix},\\
B=\begin{pmatrix}
u\b{v}+\h{u} \hb{u}-\alpha \big( u\h{u}+\hb{u}
\b{v} \big) +\delta_{{1}}\delta_{{2}} \big( {\alpha}^{2}-1
\big) \h{u} \b{v}\\
v\b{u}+\h{v} \hb{v}-\alpha
\big( \b{u} \hb{v}+v\h{v} \big) +\delta_{{1}}\delta_{{2}
} \big( {\alpha}^{2}-1 \big) \h{v} \b{u}
\end{pmatrix}=\begin{pmatrix}0\\0\end{pmatrix},
\\
Q^*=\begin{pmatrix}
\delta_{{1}}\b{v} \tb{v}+\b{v} \thb{v} \delta_{{2}}+
\hb{u} \b{v}+\thb{v} \tb{v}\\
\delta_{{1}}\b{u} \tb{u}
+\b{u} \thb{u} \delta_{{2}}+\b{u} \hb{v}+\thb{u}
 \tb{u}
\end{pmatrix}=\begin{pmatrix}0\\0\end{pmatrix},
\\
A^*=\begin{pmatrix}
\delta_{{1}}\alpha \h{u} \th{u}+\h{u} \thb{v} \delta_{{2}}
+\h{u} \hb{u}+\thb{v} \th{u}\\
\delta_{{1}}\alpha \h{v}
 \th{v}+\h{v} \thb{u} \delta_{{2}}+\thb{u} \th{v}+\h{v} \hb{v}
\end{pmatrix}=\begin{pmatrix}0\\0\end{pmatrix},
\\
B^*=\begin{pmatrix}
\t{v} \tb{v}+\thb{v} \th{u}-\alpha \big( \t{v}
\th{u}+\thb{v} \tb{v} \big)\\
\t{u} \tb{u}+\thb{u}
 \th{v}-\alpha \big( \t{u} \th{v}+\thb{u} \tb{u}
\big)
\end{pmatrix}=\begin{pmatrix}0\\0\end{pmatrix},
\end{gather*}
where $\alpha$, $\delta_1$, $\delta_2$ are parameters. It is consistent around the cube.

\begin{Remark} As in the scalar case, one can not straightforwardly impose the cube system \eqref{consist-2} on the $\mathbb{Z}^3$ lattice. It needs to be accompanied
by 7 other cube systems which are obtained from the original one by reflections.
If $R_i$ denotes a reflection in the $i$th direction, e.g., application of~$R_1$
gives the cube system depicted in Fig.~\ref{refl}, then on the cube with center
$\big(n+\frac12,m+\frac12,l+\frac12\big)$ one should impose the cube system reflected by $R_1^nR_2^mR_3^l$.
\end{Remark}

\begin{figure}[th]
\setlength{\unitlength}{0.08em}
\hspace{6cm}\begin{picture}(200,170)(10,-20)
 \put(100, 0){\circle*{6}} \put(0 ,100){\circle*{6}}
 \put( 50, 30){\circle*{6}} \put(150,130){\circle*{6}}
 \put( 0, 0){\circle*{6}} \put(100,100){\circle*{6}}
 \put( 50,130){\circle*{6}} \put(150, 30){\circle*{6}}
 \put( 0, 0){\line(1,0){100}}
 \put( 0,100){\line(1,0){100}}
 \put(50,130){\line(1,0){100}}
 \put( 0, 0){\line(0,1){100}}
 \put(100, 0){\line(0,1){100}}
 \put(150,30){\line(0,1){100}}
 \put( 0,100){\line(5,3){50}}
 \put(100,100){\line(5,3){50}}
 \put(100, 0){\line(5,3){50}}
 \dashline{3}(50,30)(0,0)
 \dashline{3}(50,30)(150,30)
\dashline{3}(50,30)(50,130)
 \dashline{2}(25, 65)(125, 65)
 \dashline{2}(75, 15)(75, 115)
 \dashline{2}(50, 50)(100, 80)
 \put(25, 65){$B^*$}
 \put(125, 65){$B$}
 \put(75, 15){$Q$}
 \put(75, 115){$Q^*$}
 \put(50, 50){$A$}
 \put(100, 80){$A^*$}
 \put(-10,-10){$T_{k_1}u$}
 \put(90,-12){$ \t u $}
 \put(40,17){$T_{k_4} \h u $}
 \put(-20,106){$T_{k_5} \b u $}
 \put(155,20){$T_{k_2} \th u $}
 \put(45,135){$T_{k_7} \hb u $}
 \put(105,92){$T_{k_3} \tb u $}
 \put(152,130){$T_{k_6} \thb u $}
\end{picture}
\caption{CAC system \eqref{consist-2} reflected in $\sim$ direction.} \label{refl}
\end{figure}
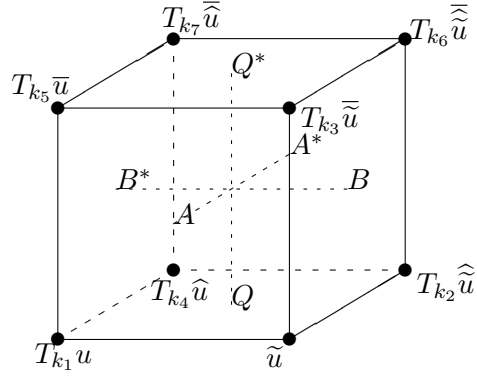

\subsection{Multi-component CAC lattice systems}\label{sec-2-2}

In this section we consider CAC lattice systems. The difference with the previous section is that we now require that the lattice equation $Q=0$ can be consistently imposed on the entire $\mathbb{Z}^2$ lattice, together with the cubes they are part of. The consequence of this requirement is two-fold:
\begin{itemize}\itemsep=0pt
\item we have to restrict ourselves to cube systems with $A=A^*$ and $B=B^*$,
\begin{subequations}\label{consist-10}
\begin{alignat}{3}
&Q\big(u,\t u,\h u, \th u\big)=0,\qquad && Q^*\big(\b u,\tb u,\hb u, \thb u\big)=0,&\\
&A\big(u,\t u,\b u, \tb u\big)=0,\qquad && A\big(\h u,\th u,\hb u, \thb u\big)=0,&\\
&B\big(u,\h u,\b u,\hb u\big)=0,\qquad && B\big(\t u,\th u,\tb u,\thb u\big)=0,&
\end{alignat}
\end{subequations}
because consistent cubes with $Q=0$ on the bottom face need to be glued together so that their common faces carry same equation, $A=0$ or $B=0$. Note that we want to allow for the possibility that $Q\neq Q^*$, so that (non-auto) B\"acklund transformations are included in the same framework.
\item we need to restrict the values the parameters of the extension, $k_i$, can acquire.
\end{itemize}

\begin{Theorem}\label{T1}
Suppose that the system \eqref{consist-10} is CAC in the sense $\thb u$ is uniquely determined by~\eqref{F} in terms of initial values $u$, $\t u$, $\h u$, $\b u$.
Extending $u$ to be a diagonal matrix \eqref{u-form}, the system
\begin{subequations}\label{consist-2'}
\begin{alignat}{3}
&Q\big(u, T_{a} \t u, T_{b} \h u, T_{a+b}\th u\big)=0,
\qquad && Q^*\big(T_{c} \b u, T_{a+c} \tb u, T_{b+c} \hb u, T_{a+b+c} \thb u\big)=0,& \\
&A\big(u, T_{a} \t u, T_{c} \b u, T_{a+c} \tb u\big)=0,
\qquad && A\big(T_{b} \h u, T_{a+b}\th u, T_{b+c} \hb u, T_{a+b+c} \thb u\big)=0,&\label{consist-2b'}\\
&B\big(u, T_{c} \b u, T_{b} \h u, T_{b+c} \hb u\big)=0,
\qquad && B\big(T_{a} \t u, T_{a+c}\tb u, T_{a+b} \th u, T_{a+b+c} \thb u\big)=0& \label{consist-2c'}
\end{alignat}
\end{subequations}
is CAC as well, where $a,b,c\in \mathbb{Z}$ $({\rm mod}~N)$, {\em and}
can be consistently defined on $\mathbb{Z}^2\otimes\{0,1\}$.
\end{Theorem}

\begin{proof}The CAC property follows directly from Lemma~\ref{T0}. We have to
show that we can consistently define the same cube system on neighboring
cubes.

\begin{figure}[ht]\centering
\setlength{\unitlength}{0.07em}
\hskip 2.0cm
\begin{picture}(200,170)(10,-20)
 \put(100, 0){\circle*{6}} \put(0 ,100){\circle*{6}}
 \put( 50, 30){\circle*{6}} \put(150,130){\circle*{6}}
 \put( 0, 0){\circle*{6}} \put(100,100){\circle*{6}}
 \put( 50,130){\circle*{6}} \put(150, 30){\circle*{6}}
 \put( 0, 0){\line(1,0){100}}
 \put( 0,100){\line(1,0){100}}
 \put(50,130){\line(1,0){100}}
 \put( 0, 0){\line(0,1){100}}
 \put(100, 0){\line(0,1){100}}
 \put(150,30){\line(0,1){100}}
 \put( 0,100){\line(5,3){50}}
 \put(100,100){\line(5,3){50}}
 \put(100, 0){\line(5,3){50}}
 \dashline{3}(50,30)(0,0)
 \dashline{3}(50,30)(150,30)
\dashline{3}(50,30)(50,130)
 \dashline{2}(25, 65)(125, 65)
 \dashline{2}(75, 15)(75, 115)
 \dashline{2}(50, 50)(100, 80)
 \put(25, 65){$B$}
 \put(125, 65){$B$}
 \put(75, 15){$Q$}
 \put(75, 115){$Q^*$}
 \put(50, 50){$A$}
 \put(100, 80){$A$}
 \put(-10,-10){$u$}
 \put(90,-15){$T_{a} \t u$}
 \put(25,30){$T_{b} \h u$}
 \put(-10,110){$T_{c} \b u $}
 \put(155,20){$T_{a+b} \th u $}
 \put(45,135){$T_{b+c} \hb u $}
 \put(105,92){$T_{a+c} \tb u $}
 \put(155,130){$T_{a+b+c} \thb u $}
\end{picture}
\hspace{.1cm}
\begin{picture}(200,170)(10,-20)
 \put(100, 0){\circle*{6}} \put(0 ,100){\circle*{6}}
 \put( 50, 30){\circle*{6}} \put(150,130){\circle*{6}}
 \put( 0, 0){\circle*{6}} \put(100,100){\circle*{6}}
 \put( 50,130){\circle*{6}} \put(150, 30){\circle*{6}}
 \put( 0, 0){\line(1,0){100}}
 \put( 0,100){\line(1,0){100}}
 \put(50,130){\line(1,0){100}}
 \put( 0, 0){\line(0,1){100}}
 \put(100, 0){\line(0,1){100}}
 \put(150,30){\line(0,1){100}}
 \put( 0,100){\line(5,3){50}}
 \put(100,100){\line(5,3){50}}
 \put(100, 0){\line(5,3){50}}
 \dashline{3}(50,30)(0,0)
 \dashline{3}(50,30)(150,30)
\dashline{3}(50,30)(50,130)
 \dashline{2}(25, 65)(125, 65)
 \dashline{2}(75, 15)(75, 115)
 \dashline{2}(50, 50)(100, 80)
 \put(25, 65){$B$}
 \put(125, 65){$B$}
 \put(75, 15){$Q$}
 \put(75, 115){$Q^*$}
 \put(50, 50){$A$}
 \put(100, 80){$A$}
 \put(-10,-10){$\t u$}
 \put(90,-15){$T_{a} \tt u$}
 \put(25,30){$T_{b} \th u$}
 \put(-10,110){$T_{c} \tb u $}
 \put(155,20){$T_{a+b} \tth u $}
 \put(45,135){$T_{b+c} \thb u $}
 \put(105,92){$T_{a+c} \ttb u $}
 \put(155,130){$T_{a+b+c} \tthb u $}
\end{picture}
\vskip 0.5cm
\caption{Neighboring cubes supporting the same cube system.}\label{neighbors}
\end{figure}
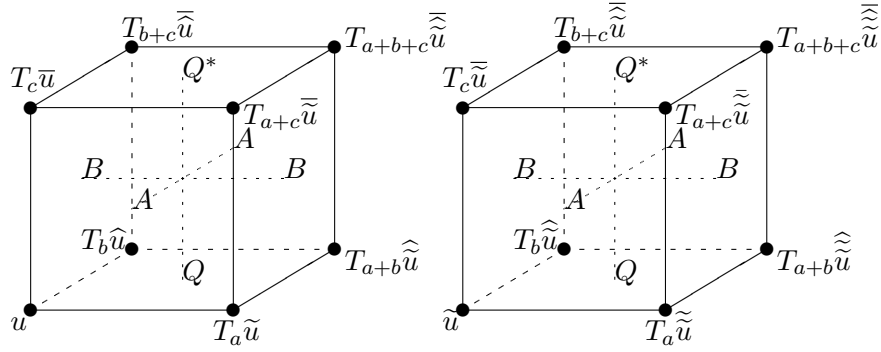

Consider two neighboring cubes, as in Fig.~\ref{neighbors}. Before we can glue them together we need to apply $T_a$ to every vertex of the cube on the right. But then we have to establish, e.g., that the shifted system of equations
\begin{gather} \label{soe1}
{\cal S}_n Q\big(u, T_a \t u, T_b \h u, T_{a+b}\th u\big) = Q\big(\t u, T_a \t{\t u}, T_b \th u, T_{a+b}\th {\t u}\big)=0,
\end{gather}
where ${\cal S}_nf(n,m)=f(n+1,m)$, is the same system of equations as the system
\begin{gather} \label{soe2}
Q\big(T_a \t u, T_{2a} \t{\t u}, T_{a+b} \th u, T_{2a+b}\th {\t u}\big)=0.
\end{gather}
Indeed, we have the identity
\[
T_a Q\big(\t u, T_a \t{\t u}, T_b \th u, T_{a+b}\th {\t u}\big) = Q\big(T_a \t u, T_{2a} \t{\t u}, T_{a+b} \th u, T_{2a+b}\th {\t u}\big),
\]
which shows that the system (\ref{soe2}) is just a rearrangement of the components of the system~(\ref{soe1}). In fact, it can be shown that for any fractional affine linear function $g$ of diagonal $N\times N$ matrices $(m_1,\ldots,m_n)$,
we have
\begin{gather*}
\theta g(m_1,\ldots,m_n)\theta^{-1}=g\big(\theta m_1\theta^{-1},\dots, \theta m_n \theta^{-1}\big),
%\label{sigma-pro}
\end{gather*}
for any invertible $N\times N$ matrix $\theta$.
\end{proof}

It follows from the proof that the equations in~(\ref{consist-2'}) on the right may be simplified, i.e.,
\begin{gather*}
Q^*\big(\b u, T_{a} \tb u, T_{b} \hb u, T_{a+b} \thb u\big)=0,\qquad
A\big(\h u, T_{a}\th u, T_{c} \hb u, T_{a+c} \thb u\big)=0,\qquad
B\big(\t u, T_{c}\tb u, T_{b} \th u, T_{b+c} \thb u\big)=0.
\end{gather*}

It also follows that in the case where $Q^*=Q$, the cube system can be imposed on the entire $\mathbb{Z}^3$-lattice. In the case where $Q^*\neq Q$ one needs a second cube system obtained from the first by the reflection $R_3$, and impose the reflected system on cubes with center $\big(n+\frac12,m+\frac12,l+\frac12\big)$ with $l$ odd.

\begin{Remark} The cubes in Figs.~\ref{F:1}(b), \ref{refl}, and \ref{neighbors} are useful to define the equations which live on the faces of the cubes. However, one should be aware that the field $u$ (which provides the support for equations~(\ref{consist-2}) and~(\ref{consist-2'}) is defined in the usual way, namely $\t u(n,m,l)=u(n+1,m,l)$, as in Fig.~\ref{F:1}. We do {\em not} have $\t u(n,m,l)=T_a u(n+1,m,l)$.
\end{Remark}

\subsubsection{Examples} \label{sec-2-4}
The $N$-component equation
\begin{gather} \label{QQ}
Q\big(u, T_a \t u, T_b \h u, T_{a+b}\th u\big)=0,
\end{gather}
where $u$ is an $N\times N$ diagonal matrix \eqref{u-form}), will be referred to as the $N[a,b]$ extension of the scalar equation \eqref{Q}. Similarly, the multi-component cube system \eqref{consist-2'} will be referred to as the $N[a,b,c]$ extension of the cube system~\eqref{consist-10}. In this terminology, the trivial Toeplitz extension introduced in Appendix~B of~\cite{Jat} corresponds to the $N[1,1]$ extension.

\subsubsection*{Discrete Burgers}
A simple example of a CAC scalar equation is the 3-point discrete Burgers equation \cite{CZZ-arxiv-2019,Levi-Bur-1983}
\begin{gather} \label{burg}
\th u (p-q+\h u-\t u)=p \h u-q \t u.
\end{gather}

The parameters in this equation, $p$, $q$, are called lattice parameters, $p$
corresponds to the tilde-direction, $q$ corresponds to the hat-direction,
and there is a third parameter, $r$, which corresponds to the bar-direction. The equations on the faces of the corresponding consistent cube are each of the form~(\ref{burg}) with different dependence on the lattice parameters, i.e., setting
\begin{gather*}
Q\big(\t u,\h u, \th u\big)=Q\big(\t u,\h u, \th u; p, q\big):= \th u (p-q+\h u-\t u)-p \h u + q \t u, \\
A\big(\t u,\b u, \tb u\big)=Q\big(\t u,\b u, \tb u; p, r\big), \qquad
B\big(\h u,\b u,\hb u\big)=Q\big(\h u,\b u,\hb u; q, r\big),
\end{gather*}
the cube system
\begin{alignat*}{3}
&Q\big(\t u,\h u, \th u\big)=0,\qquad && Q\big(\tb u,\hb u, \thb u\big)=0,& \\
&A\big(\t u,\b u, \tb u\big)=0,\qquad && A\big(\th u,\hb u, \thb u\big)=0,& \\ %\label{cubesyst}\\
&B\big(\h u,\b u,\hb u\big)=0,\qquad && B\big(\th u,\tb u,\thb u\big)=0,&
\end{alignat*}
is CAC (with no dependence on~$u$).

The $2[0,1]$ extension of equation \eqref{burg} is
\begin{gather}\label{210B}
\th u (p-q+\h v-\t u)=p \h v-q \t u,\qquad
\th v (p-q+\h u-\t v)=p \h u-q \t v,
\end{gather}
and its $2[1,1]$ extension is
\begin{gather*}%\label{211B}
\th u (p-q+\h v-\t v)=p \h v-q \t v,\qquad
\th v (p-q+\h u-\t u)=p \h u-q \t u.
\end{gather*}
The $2[1,0]$ extension is the same as (\ref{210B}) (after interchanging the two equations).

\subsubsection*{ABS equations}
The ABS equations also depend on lattice parameters. They are scalar equations of the form
\[
Q\big(u,\t u, \h u, \th u; p,q\big)=0
\]
and can be embedded in a CAC system as follows
%\begin{subequations}\label{Qsystem}
\begin{alignat*}{3}
&Q\big(u,\t u,\h u, \th u; p,q\big)=0,\qquad &&Q\big(\b u,\tb u,\hb u, \thb u; p,q\big)=0,&\\
&Q\big(u,\t u,\b u, \tb u; p,r\big)=0,\qquad &&Q\big(\h u,\th u,\hb u, \thb u; p,r\big)=0,&\\
&Q\big(u,\h u,\b u,\hb u; q,r\big)=0,\qquad &&Q\big(\t u,\th u,\tb u,\thb u; q,r\big)=0.&
\end{alignat*}

Similar to the discrete Burgers equation, each ABS equation has two types of 2-component generalizations, $2[0,1]$ and $2[1,1]$, which are respectively given by
\begin{gather}\label{2c1}
 Q\big(u,\t u, T \h u, T \th u; p,q\big)=0,
\end{gather}
and
\begin{gather}\label{2c2}
Q\big(u, T \t u, T \h u, \th u; p,q\big)=0.
\end{gather}
The latter form has been investigated in \cite{FZ14}. Explicitly, for the H1 equation, also known as the lattice potential KdV equation (lpKdv),
\begin{gather}\label{H1-1c}
\big(u-\th u\big)\big(\t u-\h u\big)+q-p=0,
\end{gather}
the $2[0,1]$ extension is
\begin{gather}\label{2cH11}
\big(u-\th v\big)\big(\t u-\h v\big)+q-p=0, \qquad
\big(v-\th u\big)\big(\t v-\h u\big)+q-p=0,
\end{gather}
and the $2[1,1]$ extension is
\begin{gather*}%\label{2cH12}
\big(u-\th u\big)\big(\t v-\h v\big)+q-p=0, \qquad
\big(v-\th v\big)\big(\t u-\h u\big)+q-p=0.
\end{gather*}
The latter appeared in \cite{BHQK12}, where the CAC property was used to construct its Lax pair. The $2[1,1]$ extension of the lattice Schwarzian KdV equation was given in \cite[equation~(B.6)]{Jat}.

When $N=3$ one can have $3[0,1]$, $3[0,2]$, $3[1,1]$, and $3[1,2]$ generalizations. For example, the $3[1,1]$ extension of H1 is
\begin{gather}
\big(u-\th w\big)\big(\t v-\h v\big)+q-p=0, \qquad
\big(v-\th u\big)\big(\t w-\h w\big)+q-p=0,\nonumber\\
\big(w-\th v\big)\big(\t u-\h u\big)+q-p=0.\label{3cH1}
\end{gather}

\begin{Remark} While each ABS-equation is part of a CAC system which comprises copies of the same equation (with appropriate dependence on the lattice parameters), this is not necessarily the case for their multi-component extensions. For example, the 2-component~($Q$) equation~(\ref{2cH11}) sits in the cube system $2[0,1,0]$ together with
\begin{gather} \label{JH1}
A\colon \ \begin{array}{l}
(u-\tb u)(\t u-\b u)+r-p=0, \\
(v-\tb v)(\t v-\b v)+r-p=0,
\end{array}\\
B\colon \ \begin{array}{l}
(u-\hb v)(\h v-\b u)+r-q=0, \\
(v-\hb u)(\h u-\b v)+r-q=0.
\end{array} \label{JH2}
\end{gather}
Here the $B$ equation has the same form as the $Q$ equation, but the $A$ equation is decoupled.
\end{Remark}

\subsubsection*{Auto-B\"acklund transformations}
There also exist CAC systems containing two different equations. For example, one can compose a CAC system using the
lattice potential modified KdV (lpmKdV, or H$3^0$) equation
\begin{gather}\label{H30}
B\big(u,\h u,\b u,\hb u;q,r\big)=q\big(u\h u-\b u\hb u\big)-r\big(u\b u-\h u\hb u\big),
\end{gather}
on the side faces and the discrete sine-Gordon (dsG) equation
\begin{gather*}%\label{dDG}
Q\big(u,\t u,\h u,\th u;p,q\big)=p\big(u\th u-\t u\h u\big)-q\big(u\t u\h u \th u-1\big),
\end{gather*}
for the other four faces of the cube $(Q,A)$ \cite{RB11,HZ-2009-preprint}.
Multi-component extension yields the CAC system
\begin{subequations}\label{consist-SG}
\begin{gather}
Q\big(u, T_{a} \t u , T_{b}\h u,T_{a+b} \th u ;p,q\big)=0,\label{consist-SG1}\\
Q\big(u,T_{a} \t u ,T_{c}\b u, T_{a+c}\tb u;p,r\big)=0,\label{consist-SG2}\\
B\big(u,T_{b}\h u,T_{c}\b u,T_{b+c}\hb u;q,r\big)=0,\label{consist-SG3}
\end{gather}
\end{subequations}
with their shifts.
In particular, we mention the $2[1,1]$ extension of the dsG equation
\begin{gather} \label{2csG}
p\big(u\th u-\t v\h v\big)-q\big(u\t v\h v \th u-1\big)=0,\qquad
p\big(v\th v-\t u\h u\big)-q\big(v\t u\h u \th v-1\big)=0,
\end{gather}
whose auto-BT consists of the dsG equation and the lpmKdV equation, which gives rise to an asymmetric Lax-pair, given in Section~\ref{slax}.

The auto-B\"acklund transformations given in Table~2 of~\cite{Atk08} provide other instances of the same situation. For example, one can take the ABS equation called Q$1^1$ as the $Q$-equation, that is
\[
Q\big(u,\t u,\h u,\th u; p,q\big):= p \big(u-\h u\big) \big(\t u-\h{\t u}\big)- q \big(u-\t u\big) \big(\h u-\h{\t u}\big)+ p q (p-q) =0.
\]
at the bottom face, and $Q\big(\b u,\tb u,\hb u,\thb u; p,q\big)=0$ on the top face. A CAC system is obtained by placing the auto-BT, where the B\"acklund parameter $r$ plays the role of the lattice parameter,
\begin{gather} \label{BT}
A\big(u,\t u,\b u,\tb u,p,r\big):= \big(u-\t u\big) \big(\b u-\tb u\big)+p \big(u+\t u+\b u +\tb u+ p +2r\big) =0,
\end{gather}
on the front face, $A\big(\h u,\th u,\hb u,\thb u,p,r\big)=0$ on the back face,
$A\big(u,\h u,\b u,\hb u,q,r\big)=0$ on the left face and $A\big(\t u,\th u,\tb u,\thb u,q,r\big)=0$ on the right face. Such CAC lattice systems can be consistently extended to multi-component CAC lattice systems by virtue of Theorem~\ref{T1}.

Thus, when $Q^*=Q$ the multi-component equations
\begin{subequations}\label{BT3}
\begin{gather}
 A\big(u, T_{a} \t u, T_{c} \b u, T_{a+c} \bt u; p,r \big)=0, \label{BT3a}\\
 B\big(u, T_{c} \b u, T_{b} \h u, T_{b+c} \bh u; r,q \big)=0, \label{BT3b}
\end{gather}
\end{subequations}
can be interpreted as an auto-BT, mapping one solution $u$ to another solution~$\b u$. This is because the top equation in \eqref{consist-2'} can be rewritten as
\[
T_c Q \big( \b u, T_{a} \bt u, T_{b} \bh u, T_{a+b} \th{\b u}; p,q \big) =0.
\]

We remark that the equation~(\ref{BT}) (in fact, any auto-BT) is an integrable equation on the~$\mathbb{Z}^2$ lattice. The equation~(\ref{BT}) is not in the ABS list and neither is the sine-Gordon equation, because they are not CAC with copies of themselves. However, they do posses an auto-BT and hence (non-symmetric) Lax pairs (where the B\"acklund parameter provides the so called spectral parameter) can be constructed, see Section~\ref{slax}.

\subsubsection*{B\"acklund transformations}
For (non-auto) B\"acklund transformations, such as the ones given in Table~3 in~\cite{Atk08}, the equations on the bottom face, $Q$, and on the top face $Q^*$ are different. For example, taking $Q=$H2,
\begin{gather*}%\label{H2}
\big({u}-\th {u}\big)\big(\t{{u}}-\h{{u}}\big)-(p-q)\big({u}+\t{{u}}+\h{{u}}+\th {u}+p+q\big)=0,
\end{gather*}
and posing
%\begin{subequations}\label{nBT}
\begin{gather*}
 A \colon \ u+\t u+p=2{\b u}\t {\b u}, \qquad %\label{nBT-a}\\
 B\colon \ u+\h u+q=2{\b u}\h {\b u}, %\label{nBT-b}
\end{gather*}
and their shifted versions, $\h{A}$ and $\t{B}$ on the side faces, one finds that on the top face the variable ${\b u}$ satisfies the $Q^*=$H1 equation \eqref{H1-1c}.

In the general $N$-component cube system, denoted $N[a,b,c]$, the side system
%\begin{subequations}\label{nBT'}
\begin{gather*}
 u+T_a \t u+p=2(T_c {\b u})\big( T_{a+c} \t {\b u}\big),\qquad %\label{nBT-a'}\\
 u+T_b \h u+q=2(T_c {\b u}) \big(T_{b+c} \h {\b u}\big) %\label{nBT-b'}
\end{gather*}
provides a BT between the $N$-component H2 system
\begin{gather*}%\label{H2-Nc}
\big({u}-T_{a+b}\th {u}\big)\big(T_a\t{{u}}-T_b\h{{u}}\big)-(p-q)\big({u}+T_a\t{{u}}+T_b\h{ u}+T_{a+b}\th u+p+q\big)=0,
\end{gather*}
and the $N$-component H1 system
\begin{gather*}%\label{H1-Nc}
 \big(T {\b u}-T_{a+b} \th {\b u}\big)\big(T_{a}\t{{\b u}}-T_{b} \h{{\b u}}\big)-p+q=0
\end{gather*}

There are other examples of CAC lattice systems such as the ones in~\cite{Hietarinta-JNMP-2019}. These all allow multi-component extension.

\subsubsection{Lax pairs} \label{slax}
In a consistent system of the form \eqref{consist-2'} the Lax pair of equation~\eqref{QQ}
can be constructed through the BT \eqref{BT3} following the standard procedure,
cf.~\cite{ABS03,BHQK12,Nij02}. Here one would introduce $\b u= g f^{-1}$ with
\begin{gather}\label{fg}
f=\operatorname{diag}(f_1,f_2,\dots,f_N),\qquad g=\operatorname{diag}(g_1,g_2,\dots,g_N),
\end{gather}
and $\Psi=(f_1,f_2,\dots,f_N, g_1, g_2, \dots, g_N)^{\rm T}$,
and then equation \eqref{BT3} can be cast into the form
$\t \Psi=L(u,\t u)\Psi$, $\h \Psi=M(u,\h u)\Psi$.
With the correct scaling factors, the pair of matrices $L$, $M$ form a $2N\times 2N$ Lax pair of~\eqref{QQ}. In Appendix~\ref{appendixA} we show how this approach yields \eqref{Lax-LM2}.

On the other hand, one can directly write down a Lax pair of the $N$-component system~\eqref{QQ} in terms of Lax matrices of the scalar equation.
Suppose that~\eqref{consist-10} is a scalar 3D consistent lattice system,
and the bottom equation
\begin{gather}\label{Q0}
Q\big(u,\t u,\h u,\th u\big)=0
\end{gather}
has a $2\times 2$ Lax pair (e.g., the one obtained from the BT at hand)
\begin{gather}
\t \psi=\mathcal{L}(u,\t u)\psi,\qquad \h \psi=\mathcal{M}(u,\h u)\psi,\label{Lax-LM1}
\end{gather}
where $\psi=(\psi_1,\psi_2)^{\rm T}$ and $\mathcal{L}$ and $\mathcal{M}$ are $2\times 2$ matrices. Considering $N$ copies of the equation, each scalar equation
$Q\big(u_i,\t u_i,\h u_i,\th u_i\big)=0$ has a $2\times 2$ Lax pair $\mathcal{L}(u_i,\t u_i)$, $\mathcal{M}(u_i,\h u_i)$. Then, we have the following Lax pair for the multi-component case.

\begin{Theorem}\label{T4}
Suppose that the scalar equation \eqref{Q0}
has a Lax pair \eqref{Lax-LM1}. Then the $N$-component extension \eqref{QQ} has a $2N\times 2N$ Lax pair
\begin{gather}\label{Lax-LM2}
\t \Phi=\theta^{-a-c}L(u, T_a\t u)\theta^c\Phi,\qquad \h \Phi=\theta^{-b-c}M(u, T_b\h u )\theta^c\Phi.
\end{gather}
where, with $u$ the diagonal matrix \eqref{u-form} and $v=\operatorname{diag}(v_1, v_2, \dots, v_N)$,
\begin{subequations}\label{LM}
\begin{gather}
 L(u,v)=\operatorname{diag}(\mathcal{L}(u_1,v_1),\mathcal{L}(u_2,v_2),\dots,\mathcal{L}(u_N,v_N)),\\
 M(u,v)=\operatorname{diag}(\mathcal{M}(u_1,v_1),\mathcal{M}(u_2,v_2),\dots,\mathcal{M}(u_N,v_N))
\end{gather}
\end{subequations}
in which $\mathcal{L}$ and $\mathcal{M}$ are the Lax matrices given in~\eqref{Lax-LM1},
$\theta=\sigma_{2N}^2$ and $\sigma_{2N}$ is a $2N\times 2N$ cyclic matrix defined by~\eqref{sigma}.
\end{Theorem}

Note that the gauge transformation
$\Phi'= \theta^c\Phi$,
transforms the Lax pair \eqref{Lax-LM2} into
\begin{gather}\label{Lax-LM2'}
\t \Phi'=\theta^{-a}L(u, T_a\t u) \Phi',\qquad \h \Phi'=\theta^{-b}M(u, T_b\h u ) \Phi',
\end{gather}
which does not depend on $c$.

\begin{proof}The compatibility of the linear system $\t \Phi=L(u, \t u)\Phi$, $\h \Phi=M(u, \h u ) \Phi$
equals
\begin{gather}\label{LM-ML}
 L\big(\h u,\th u\big)M (u,\h u )=M\big(\t u,\th u\big)L(u,\t u),
\end{gather}
which is equivalent to equation \eqref{Q0} with diagonal matrix $u$ given by \eqref{u-form}.
From the compa\-ti\-bility of \eqref{Lax-LM2} we find
\begin{gather*}
 \theta^{-a-c} L\big(\h u,T_{a}\th u\big)\theta^{c}\theta^{-b-c}M(u,T_b\h u)\theta^c
 =\theta^{-b-c} M\big(\t u,T_{b}\th u\big)\theta^{c}\theta^{-a-c}L(u, T_a\t u)\theta^c\\
\qquad{} \Longleftrightarrow \theta^b L\big(\h u,T_{a}\th u\big)\theta^{-b}M(u,T_b\h u) =\theta^a M\big(\t u,T_{b}\th u\big)\theta^{-a}L(u, T_a\t u)\\
\qquad {} \Longleftrightarrow L\big(T_b\h u,T_{a+b}\th u\big)M(u,T_b\h u)=M\big(T_a\t u,T_{a+b}\th u\big)L(u, T_a\t u),
\end{gather*}
which gives rise to \eqref{QQ}, as \eqref{Q0} arises from~\eqref{LM-ML}.
\end{proof}

As an example, consider the H1 equation \eqref{H1-1c} which admits the Lax pair
\begin{gather*} \label{LMH}
\mathcal{L}(u,\t u)=\mathcal{H}(u,\t u,p),\qquad \mathcal{M}(u,\h u)=\mathcal{H}(u,\h u,q),
\end{gather*}
with
\begin{gather*}
\mathcal{H}(u,\t u,p)=
\left(
 \begin{matrix}
 -u & u\t u+p-r \\
 -1 & \t u
 \end{matrix}
\right).
\end{gather*}
According to \eqref{Lax-LM2'}, the $3[1,1]$ extension of H1 \eqref{3cH1} admits the Lax pair
\begin{gather*}
\t \Phi=
\left(
 \begin{matrix}
 0 & 0 & \mathcal{H}(w,\t u,p) \\
 \mathcal{H}(u,\t v,p) & 0 & 0 \\
 0 & \mathcal{H}(v,\t w,p) & 0
 \end{matrix}
\right)\Phi, \\
\h \Phi=
\left(
 \begin{matrix}
 0 & 0 & \mathcal{H}(w,\h u,q) \\
 \mathcal{H}(u,\h v,q) & 0 & 0 \\
 0 & \mathcal{H}(v,\h w,q) & 0
 \end{matrix}
\right)\Phi.
\end{gather*}

\subsubsection*{Asymmetrical Lax pairs}
When equation $A$ is not related to equation $B$ by a standard change in lattice parameters, the corresponding Lax pair for $Q$ is asymmetrical. Similarly, one also finds asymmetry if one constructs a Lax pair for an auto-BT ($A$), using its auto-BT given by $B$, $Q$. For example, the auto-BT \eqref{consist-SG2}--\eqref{consist-SG3} provides an asymmetrical Lax pair for the 2-component dsG equation~\eqref{2csG},
\[
\t \Phi=
\left(
 \begin{matrix}
 0 & \mathcal{L}(u,\t v) \\
 \mathcal{L}(v,\t u)& 0
 \end{matrix}
\right)\Phi, \qquad \h \Phi=
\left(
 \begin{matrix}
 0 & \mathcal{M}(u,\h v) \\
 \mathcal{M}(v,\h u)& 0
 \end{matrix}
\right)\Phi,
\]
where $a=b=c=1$ and
\[
\mathcal{L}(u,\t u)=
\left(
 \begin{matrix}
 p& -r\t u\\
 -\dfrac{r}{u}& \dfrac{p\t u}{u} \\
 \end{matrix}
\right),\qquad
\mathcal{M}(u,\h u)=
\left(
 \begin{matrix}
 \dfrac{r\h u}{u}& -\frac{q}{u}\\
 -q\h u& r
 \end{matrix}
\right).
\]

For all the auto-BTs gives in \cite[Table 2]{Atk08} a superposition principle emerges for solutions of the equation that are related by the auto-BT. This gives rise to a different asymmetrical Lax-pair for the auto-BT. The superposition principle for solutions of the lmpKdV equation related by the dsG equation is
\begin{gather}\label{H30s}
Q\big(u,\t u,\dot u,\dot{\h u};p,s\big)=s\big(u\t u-\dot{u}\dot{\t u}\big)-p\big(u\dot u-\t u\dot{\t u}\big),
\end{gather}
which is an lmpKdV equation with $p$ and $s$ interchanged. The cube system with the superposition principle (\ref{H30s}) on the bottom and top faces and the dsG equation on the side faces, cf.\ Fig.~\ref{SPP}, is consistent, and admits multi-component extension.
\begin{figure}[h!]\centering
\includegraphics[width=9cm]{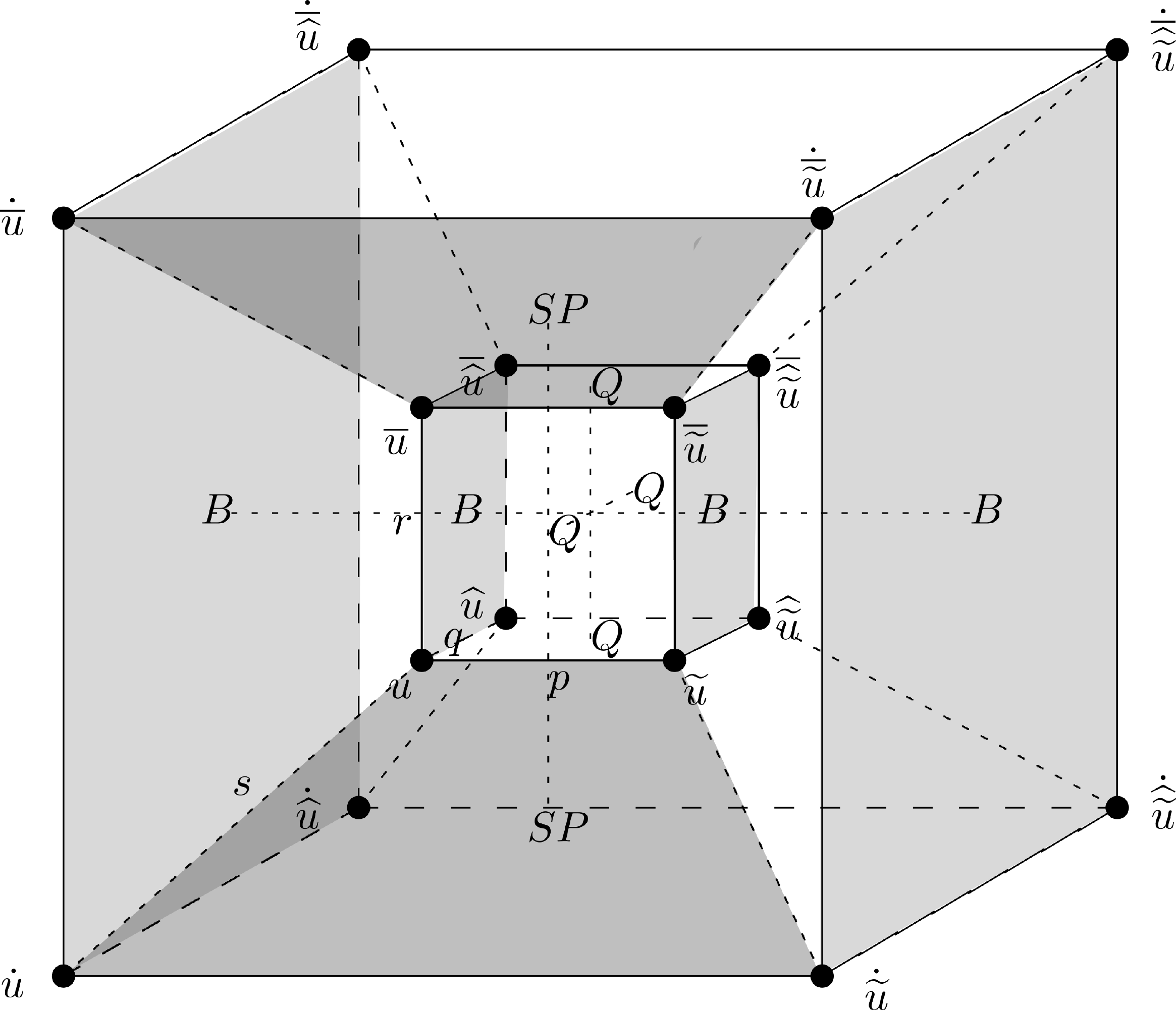}
\caption{On the inner cube we have the lmpKdV ($B$) on the left face and on the right face. They are connected by the auto-BT (with parameter $p$) which is the dsG equation ($Q$) on the front, back, bottom and top faces. On the left cube the auto-BT (in dot-direction with parameter $s$) yields an lmpKdV equation on the left face. On the right cube the auto-BT with parameter $s$ yields a fourth lmpKdV equation on the right face. The four solutions to these four lmpKdV equations are related by the superposition principle. On the front cube the superposition principle ($SP$), which is the lmpKdV equation with $p$ and $s$ interchanged, is on the top and the bottom face, and the dsG equation is on the four side faces. This provides another CAC cube system.}\label{SPP}
\end{figure}

One can also construct Lax-pairs from non-auto BTs \cite[Table 3]{Atk08}, however, these will not contain a spectral parameter.

\subsubsection[Counting $N$-component extensions of ABS lattice equations]{Counting $\boldsymbol{N}$-component extensions of ABS lattice equations}\label{sec-2-3}

\subsubsection*{D4 symmetry}
The ABS lattice equations are D4 symmetric, i.e.,
\begin{gather}\label{D4-ABS}
 Q\big(u,\t{u},\h{u},\h{\t{u}};p,q\big)=\pm Q\big(u,\h{u},\t{u},\h{\t{u}};q,p\big)=\pm Q\big(\h{u},\th u, u, \t{u};p,q\big).
\end{gather}
In (\ref{D4-ABS}) we can consider $u$ and its shifts to be diagonal matrices \eqref{u-form}. By relabelling we also have
\begin{align*}
Q\big(u,T_a\t{u},T_b\h{u},T_{a+b}\h{\t{u}};p,q\big)=& \pm Q\big(u,T_b\h{u},T_a\t{u},T_{a+b}\h{\t{u}};q,p\big)\nonumber\\
= & \pm Q\big(T_a\t{u},u,T_{a+b}\h{\t{u}},T_b\h{u};p,q\big)
=\pm Q\big(T_b\h{u},T_{a+b}\h{\t{u}},u,T_a\t{u};p,q\big).
%\label{D4}
\end{align*}

Now we introduce $v(n,m)=u(-n,m)$, and in terms of $v$ we write \eqref{QQ} as
\begin{gather*}
Q\big(v, T_{a} \dt v, T_{b}\h v, T_{a+b} \h{\dt v}; p,q\big)=0.
%\label{Qv}
\end{gather*}
Applying a tilde-shift, by virtue of D4 symmetry we obtain
\begin{align*}
 Q\big(\t v, T_{a} v, T_{b}\th v, T_{a+b} \h v; p,q\big)=0 &
 \Rightarrow Q\big(T_{a} v, \t v, T_{a+b}\h v, T_{b} \th v; p,q\big)=0\\
& \Rightarrow T_a Q\big( v, T_{-a} \t v, T_{b}\h v, T_{b-a} \th v; p,q\big)=0\\
& \Rightarrow Q\big( v, T_{-a} \t v, T_{b}\h v, T_{b-a} \th v; p,q\big)=0.
\end{align*}
Since $T_{-a}=T_{N-a}$, the above relation indicates that
$N[a,b]$ and $N[N-a,b]$ generate same $N$-component system up to reflection $R_1\colon n \leftrightarrow -n$. This leads to the following proposition.

\begin{Proposition}\label{P1}
For the ABS equation \eqref{Q}, due to D4 symmetry, the cases
\[N[a, b],\quad N[b,a], \quad N[a, N-b], \quad N[N-a, b],\quad N[N-a, N-b]\]
are all equivalent up to coordinate refections. Consequently, we can assume $0\leq a\leq b\leq c\leq N/2$ without loss of generality.
\end{Proposition}

\subsubsection*{Decoupling}
Let us consider the $4[2,2]$ extension of \eqref{Q}:
\begin{subequations}\label{4c}
\begin{gather}
 Q\big(u_1,\t{u}_3,\h{u}_3,\th u_1;p,q\big)=0, \label{4c-a}\\
 Q\big(u_2,\t{u}_4,\h{u}_4,\th u_2;p,q\big)=0, \label{4c-b}\\
 Q\big(u_3,\t{u}_1,\h{u}_1,\th u_3;p,q\big)=0, \label{4c-c}\\
 Q\big(u_4,\t{u}_2,\h{u}_2,\th u_4;p,q\big)=0. \label{4c-d}
\end{gather}
\end{subequations}
This system is decoupled into two $2[1,1]$ systems, namely (\ref{4c-a}), (\ref{4c-c}) and (\ref{4c-b}), (\ref{4c-d}).

In the next theorem, for which we include a proof in Appendix~\ref{adec}, we give conditions which decide when a system is decoupled or non-decoupled. The greatest common divisor between integers $a,b,\ldots,c$ will be denoted $\operatorname{gcd}(a,b,\ldots,c)$.

\begin{Theorem}\label{T2}
Let $d=\operatorname{gcd}(a,b,N)$. If $d>1$ the $N[a,b]$ extension \eqref{QQ} can be decomposed into $d$ sets of $s[a/d, b/d]$ form of the equation \eqref{Q}, where $s=N/d$. If $d=1$ the system is non-decoupled.
\end{Theorem}

It follows that if $N$ is prime the only decoupled case is $N[0,0]$, which corresponds to the trivial multi-component extension. For $N=6$, we have five extensions which decouple,
\[
6[0,0],\quad 6[0,2],\quad 6[0,3],\quad 6[2,2],\quad 6[3,3]\]
and five that do not decouple,
\[6[0,1],\quad 6[1,1],\quad 6[1,2],\quad 6[1,3],\quad 6[2,3].
\]
We let $\alpha_N$ denote the number of $N$-component
extensions \eqref{QQ} that decouple, and we let $\beta_N$ denote the number of $N$-component
extensions \eqref{QQ} that do not decouple. Thus, $\alpha_6=\beta_6=5$.

In the following theorem we give formulas for the functions $\alpha_N$ and $\beta_N$. We use notation as follows. Let $s$ be a set. By $\mathcal{P}(s)$ we denote the powerset of~$s$, $\#s$ we denote the number of elements in $s$, and $\Pi s$ denotes the product of the elements in $s$, e.g., with $s=\{1,2,3\}$ we have
\[
\mathcal{P}(s)=\{\varnothing,\{1\},\{2\},\{3\},\{1,2\},\{1,3\},\{2,3\},\{1,2,3\}\},\qquad
\#s=3,\qquad \Pi s=6,
\]
and $\#\varnothing=0$, $\Pi \varnothing=1$. Furthermore, for $n\in\mathbb{N}$ we denote the set of prime divisors of $n$ by $\mathbb{P}_n$, i.e., if $n$ has prime decomposition
$n=\prod\limits_{i=1}^l p_i^{m_i}$, then $\mathbb{P}_n=\{p_1,p_2,\ldots,p_l\}$.

\begin{Theorem}\label{T3}
For any given positive integer $N >1$, the numbers of decoupled and non-decoupled $N$-component ABS systems \eqref{QQ} are respectively given by
\begin{gather} \label{Tn}
 \alpha_N=\sum_{s\in\mathcal{P}(\mathbb{P}_N)\backslash \varnothing}(-1)^{\#s+1}
 \begin{pmatrix}
 \left\lfloor \frac{N}{2\Pi s}\right\rfloor+2 \\
 2
\end{pmatrix}
\end{gather}
and
\begin{gather} \label{Sn}
 \beta_N= \sum_{s\in\mathcal{P}(\mathbb{P}_N)}(-1)^{\#s}
\begin{pmatrix}
 \left\lfloor \frac{N}{2\Pi s}\right\rfloor+2 \\
 2
\end{pmatrix}
\end{gather}
where $\binom nm = \frac{n!}{m!(n-m)!}$ and $\flo\cdot$ represents the floor function.
\end{Theorem}

Formulas (\ref{Tn}) and (\ref{Sn}) yield
\begin{gather*}
\alpha_{2,3,\ldots,20}=1,\,1,\,3,\,1,\,5,\,1,\,6,\,3,\,8,\,1,\,13,\,1,\,12,\,8,\,15,\,1,\,22,\,1,\,24, \\
\beta_{2,3,\ldots,20}=2,\,2,\,3,\,5,\,5,\,9,\,9,\,12,\,13,\,20,\,15,\,27,\,24,\,28,\,30,\,44,\,33,\,54,\,42,
\end{gather*}
and we note that $\alpha_{2n}+\beta_{2n}=\alpha_{2n+1}+\beta_{2n+1}=\binom{n+2}2$, and $\alpha_p=1$ when $p$ is prime.

\begin{proof}Due to $0\leq a\leq b\leq N/2$, the total number of $N[a,b]$ extensions is
\begin{gather} \label{Tot}
\alpha_N+\beta_N=\binom{\flo{N/2}+1}{2}+\binom{\flo{N/2}+1}{1}=\binom{\flo{N/2}+2}2.
\end{gather}
First consider the decoupled case. For any $p\in\mathbb{P}_N$, if
\[
a,b\in \{0, p, 2p, 3p, \dots, \flo{N/(2p)} p\}
\]
then $p$ is a divisor of $\operatorname{gcd}(a,b,N)$. By Theorem \ref{T2} this gives rise to $\binom{\flo{N/(2p)}+2}2$ decoupled $N$-component systems. Similarly, for
$q\in\mathbb{P}_N$ with $q\neq p$ we find $\binom{\flo{N/(2q)}+2}2$ decoupled $N$-component systems. However, a number of these
systems we would have already counted, namely the $\binom{\flo{N/(2pq)}+2}2$ systems where
\[
a,b\in \{0, pq, 2pq, 3pq, \dots, \flo{N/(2pq)} pq\}.
\]
Running through all the primes in $\mathbb{P}_N$ by the inclusion-exclusion principle one finds the formula~(\ref{Tn}) for $\alpha_N$. Due to~(\ref{Tot}) the formula for $\beta_N$ is then given by~\eqref{Sn}.
\end{proof}

\subsection{Multi-component CAC 3D lattice equations}\label{sec-6}
3D lattice equations defined on a 3D cube can be consistent around a 4D cube. In affine linear case, certain 8-point and 6-point lattice equations, see Fig.~\ref{F:7}, have been verified to be CAC in this sense \cite{ABS03,ABS-IMRN-2011}.
We mention in particular the lattice AKP equation (a.k.a.\ the Hirota equation~\cite{Hir})
\begin{gather}\label{AKP}
\alpha_1 \t u\bh u + \alpha_2 \h u\bt u + \alpha_3 \b u\th u=0
\end{gather}
and the lattice BKP equation (a.k.a.\ the Miwa equation \cite{Miwa-1982})
\begin{align}\label{BKP}
\alpha_1 \t u\bh u + \alpha_2 \h u\bt u + \alpha_3 \b u\th u +\alpha_4 u\th{\b u}=0.
\end{align}
One can include arbitrary coefficients $\{\alpha_j\}$ in the equations,
which can be gauged to any nonzero value~\cite{Nimmo-JPA-1997}.
The stencils of these two equations are depicted in Fig.~\ref{F:7}.

\begin{figure}[h]\centering
\setlength{\unitlength}{0.0005in}
\begin{picture}(3482,3700)(0,-500)
%\put(-1200,800){\makebox(0,0)[lb]{$(b)$}}
\put(450,1883){\circle*{150}}
\put(-100,1883){\makebox(0,0)[lb]{${\bt u}$}}
\put(1275,2708){\circle*{150}}
\put(825,2708){\makebox(0,0)[lb]{$\b u$}}
\put(3075,2708){\circle*{150}}
\put(3375,2633){\makebox(0,0)[lb]{${\bh u}$}}
\put(2250,83){\circle*{150}}
\put(2650,8){\makebox(0,0)[lb]{${\th u}$}}
%\put(100,908){\circle*{150}}
%\put(825,908){\makebox(0,0)[lb]{$u$}}
%\put(1950,2000){\makebox(0,0)[lb]{${\th{\b u}}$}}
\put(450,83){\circle*{150}}
\put(0,8){\makebox(0,0)[lb]{$\t u$}}
\put(3075,908){\circle*{150}}
\put(3300,833){\makebox(0,0)[lb]{$\h u$}}
\drawline(1275,2708)(3075,2708)
\drawline(1275,2708)(450,1883)
\drawline(450,1883)(450,83)
\drawline(3075,2708)(450,1883)
\drawline(3075,2708)(2250,83)
%\drawline(450,1883)(2250,1883)
\drawline(450,1883)(2250,83)
\drawline(3075,2708)(3075,908)
%\dashline{60.000}(1275,908)(450,83)
\dashline{60.000}(1275,2708)(450,83)
\dashline{60.000}(1275,2708)(3075,908)
%\drawline(2250,1883)(2250,83)
\drawline(450,83)(2250,83)
\drawline(3075,908)(2250,83)
\dashline{70.000}(3075,908)(450,83)
\put(-500,-700){\small{(a) 6-point stencil of lattice AKP \eqref{AKP}.} }
\end{picture}
\hspace{2cm}
\begin{picture}(3482,3700)(0,-500)
%\put(-1200,1000){\makebox(0,0)[lb]{$(b)$}}
\put(450,1883){\circle*{150}}
\put(-100,1883){\makebox(0,0)[lb]{${\bt u}$}}
\put(1275,2708){\circle*{150}}
\put(825,2708){\makebox(0,0)[lb]{$\b u$}}
\put(3075,2708){\circle*{150}}
\put(3375,2633){\makebox(0,0)[lb]{${\bh u}$}}
\put(2250,83){\circle*{150}}
\put(2650,8){\makebox(0,0)[lb]{${\th u}$}}
\put(1275,908){\circle*{150}}
\put(825,908){\makebox(0,0)[lb]{$u$}}
\put(2250,1883){\circle*{150}}
%\put(2250,1883){\circle{220}}
%\put(2250,1883){\circle{80}}
\put(1950,2000){\makebox(0,0)[lb]{${\th{\b u}}$}}
\put(450,83){\circle*{150}}
\put(0,8){\makebox(0,0)[lb]{$\t u$}}
\put(3075,908){\circle*{150}}
\put(3300,833){\makebox(0,0)[lb]{$\h u$}}
\drawline(1275,2708)(3075,2708)
\drawline(1275,2708)(450,1883)
\drawline(450,1883)(450,83)
\drawline(3075,2708)(2250,1883)
\drawline(450,1883)(2250,1883)
\drawline(3075,2633)(3075,908)
\dashline{60.000}(1275,908)(450,83)
\dashline{60.000}(1275,908)(3075,908)
\drawline(2250,1883)(2250,83)
\drawline(450,83)(2250,83)
\drawline(3075,908)(2250,83)
\dashline{60.000}(1275,2633)(1275,908)
\put(-380,-700){{(b) 8-point stencil of lattice BKP \eqref{BKP}. }}
\end{picture}
\caption{Stencils of 3D lattice equations.}\label{F:7}
\end{figure}
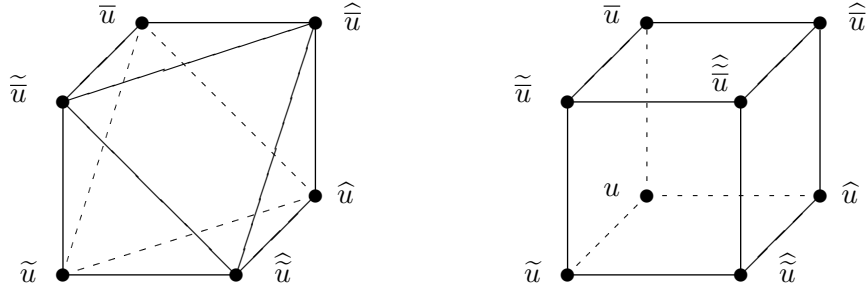

Any 3D CAC lattice equation can be generalised to a multi-component 3D equation which inherits the CAC property. We have the following result.

\begin{Theorem}\label{T5}Suppose that the scalar lattice system
\begin{alignat*}{3}
&Q\big(u,\t u,\h u, \b u, \th u, \tb u, \hb u,\thb u\big)=0,\qquad && Q\big(\dot u,\dot{\t u},\dot{\h u}, \dot{\b u}, \dot{\th u}, \dot{\tb u}, \dot{\hb u},\dot{\thb u}\big)=0,&\\
&A\big(u,\t u,\h u, \dot u, \th u, \dot{\t u}, \dot{\h u},\dot{\th u}\big)=0,\qquad && A\big(\b u,\tb u,\hb u, \dot{\b u}, \thb u, \dot{\tb u}, \dot{\hb u},\dot{\thb u}\big)=0,&\\
&B\big(u,\t u,\dot u, \b u, \dot{\t u}, \tb u, \dot{\b u},\dot{\tb u}\big)=0,\qquad &&B\big(\h u,\th u,\dot u, \hb u, \dot{\th u}, \thb u, \dot{\hb u},\dot{\thb u}\big)=0,&\\
&C\big(u,\dot u,\h u, \b u, \dot{\h u}, \dot{\b u}, \hb u,\dot{\hb u}\big)=0,\qquad && C\big(\t u,\dot{\t u},\th u, \tb u, \dot{\th u}, \dot{\tb u}, \thb u,\dot{\thb u}\big)=0&
\end{alignat*}
is consistent around the 4D cube in Fig.~{\rm \ref{F:8}}, i.e., the value of
$\dot{\thb u}$ is uniquely determined by suitably given initial values.
Then, after replacing $u$ with diagonal form \eqref{u-form}, the following system
\begin{subequations}\label{consist-h2}
\begin{gather}
 Q\big(u, T_a \t u, T_b\h u, T_c\b u, T_{a+b}\th u, T_{a+c}\tb u, T_{b+c}\hb u,T_{a+b+c}\thb u\big)=0,\label{QQ3D}\\
 Q\big(T_d \dot u, T_{a+d}\dot{\t u},T_{b+d} \dot{\h u}, T_{c+d} \dot{\b u}, T_{a+b+d}\dot{\th u}, T_{a+c+d}\dot{\tb u}, T_{b+c+d}\dot{\hb u},
 T_{a+b+c+d}\dot{\thb u}\big)=0,\\
 A\big(u,T_{a}\t u,T_{b} \h u, T_{d}\dot u,T_{a+b} \th u, T_{a+c}\dot{\t u}, T_{b+d}\dot{\h u}, T_{a+b+d}\dot{\th u}\big)=0,\\
 A\big(T_{c}\b u,T_{a+c}\tb u,T_{b+c}\hb u, T_{c+d}\dot{\b u}, T_{a+b+c}\thb u, T_{a+b+d}\dot{\tb u},T_{b+c+d} \dot{\hb u},T_{a+b+c+d}\dot{\thb u}\big)=0,\\
 B\big(u,T_{a}\t u,T_{d}\dot u, T_{c}\b u, T_{a+d}\dot{\t u}, T_{a+c}\tb u, T_{c+d}\dot{\b u},T_{a+c+d}\dot{\tb u}\big)=0,\\
 B\big(T_{b}\h u,T_{a+b}\th u,T_{d}\dot u, T_{b+c}\hb u, T_{a+b+d}\dot{\th u}, T_{a+b+c}\thb u, T_{b+c+d} \dot{\hb u},T_{a+b+c+d}\dot{\thb u}\big)=0,\\
 C\big( u,T_{d}\dot u,T_{b}\h u, T_{c}\b u, T_{b+d}\dot{\h u}, T_{c+d}\dot{\b u}, T_{b+c}\hb u,T_{b+c+d}\dot{\hb u}\big)=0,\\
 C\big(T_{a}\t u,T_{a+d}\dot{\t u},T_{a+b}\th u, T_{a+c}\tb u, T_{a+b+d}\dot{\th u},T_{a+c+d} \dot{\tb u}, T_{a+b+c}\thb u,T_{a+b+c+d}\dot{\thb u}\big)=0
\end{gather}
\end{subequations}
is also consistent around the 4D cube.
\end{Theorem}
We will call (abusing notation) equation \eqref{QQ3D} the $N[a,b,c]$ extension of the scalar equation
\begin{gather}\label{Q3D}
Q\big(u,\t u,\h u, \b u, \th u, \tb u, \hb u,\thb u\big)=0.
\end{gather}

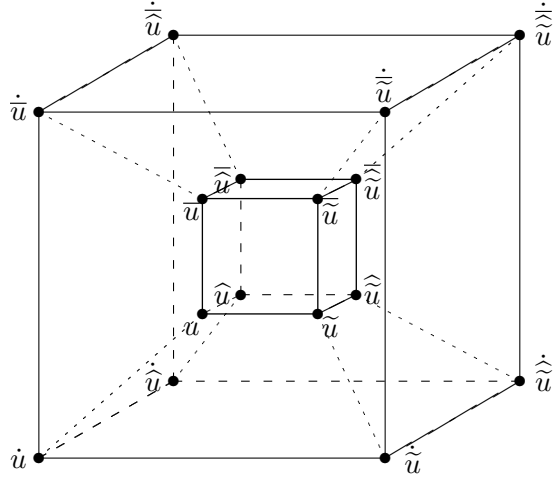
\begin{figure}\centering
\setlength{\unitlength}{0.0010in}
\hspace{-70mm}\begin{picture}(2400,2400)(-1500,-50)
 \put(1800, 0){\circle*{60}} \put(0 ,1800){\circle*{60}}
 \put( 700, 400){\circle*{60}} \put(700,2200){\circle*{60}}
 \put( 0, 0){\circle*{60}} \put(1800,1800){\circle*{60}}
 \put( 2500,400){\circle*{60}} \put(2500, 2200){\circle*{60}}
 \put( 0, 0){\line(1,0){1800}}
 \put( 0,1800){\line(1,0){1800}}
 \put(700,2200){\line(1,0){1800}}
 \put( 0, 0){\line(0,1){1800}}
 \put(1800, 0){\line(0,1){1800}}
 \put(2500,400){\line(0,1){1800}}

\drawline( 0,1800)(700,2200)
\drawline(1800,1800)(2500, 2200)
\drawline(1800, 0)(2500, 400)
\dashline{50}(0, 0)(700, 400)
 \dashline{50}(700,400)(0,0)
 \dashline{50}(700,400)(2500,400)
\dashline{50}(700,400)(700,2200)

\drawline(850, 750)(1450 ,750)
\drawline(1650, 850)(1450 ,750)
\drawline(850, 750)(850,1350)
\drawline(1450, 1350)(850,1350)
\drawline(1450, 1350)(1450,750)
\drawline(1050, 1450)(850,1350)
\drawline(1050, 1450)(1650,1450)
\drawline(1450, 1350)(1650,1450)
\drawline(1650, 850)(1650,1450)
\dashline{50}(850, 750)(1050 ,850)
\dashline{50}(1650, 850)(1050 ,850)
\dashline{50}(1050, 1450)(1050 ,850)

\put(850, 750){\circle*{60}} \put(1450 ,750){\circle*{60}}
 \put(1650, 850){\circle*{60}} \put(850,1350){\circle*{60}}
 \put(1450, 1350){\circle*{60}} \put(1650,1450){\circle*{60}}
 \put(1050, 1450){\circle*{60}} \put(1050 ,850){\circle*{60}}
\dashline{18}(1800, 0)(1450, 750)
\dashline{18}(0, 0)(850, 750)
\dashline{18}(700, 400)(1050, 850)
\dashline{18}(2500,400)(1650, 850)
\dashline{18}(2500,2200)(1650, 1450)
\dashline{18}(700, 2200)(1050, 1450)
\dashline{18}(1800, 1800)(1450, 1350)
\dashline{18}(0, 1800)(850, 1350)
 \put(-150,-50){$\dot u$}
 \put(1900, -80 ){$\dot{\t u}$}
 \put(1470, 630){$\t u$}
 \put(750, 640){$u$}

 \put(550, 350){$\dot{\h{u}}$}
 \put(910, 800){$\h u$}
 \put(2580,350){$\dot{\th u}$}
 \put(1690, 800){$\th u$}

 \put(2580,2150){$\dot{\thb{u}}$}
 \put(1690, 1350){$\thb u$}
 \put(550, 2200){$\dot{\hb u}$}
 \put(910, 1380){$\hb u$}

 \put(1750, 1850){$\dot{\tb{u}}$}
 \put(1470, 1220){$\tb u$}
 \put(-150, 1760){$\dot{\b u}$}
 \put(750, 1240){$\b u$}
\end{picture}
\caption{4D consistency around a hypercube.}\label{F:8}
\end{figure}

As examples, the $2[1,1,1]$ extension of the AKP equation \eqref{AKP} is
\begin{subequations}\label{AKP-2c}
\begin{gather}
\alpha_1 \t v\bh u + \alpha_2 \h v\bt u + \alpha_3 \b v\th u=0, \label{AKP-2c-1}\\
\alpha_1 \t u\bh v + \alpha_2 \h u\bt v + \alpha_3 \b u\th v=0, \label{AKP-2c-2}
\end{gather}
\end{subequations}
and the $2[1,1,1]$ extension of the BKP equation \eqref{BKP} is
\begin{subequations}\label{BKP-2c}
\begin{gather}
\alpha_1 \t v\bh u + \alpha_2 \h v\bt u + \alpha_3 \b v\th u+\alpha_4 u\th{\b v}=0, \label{BKP-2c-1}\\
\alpha_1 \t u\bh v + \alpha_2 \h u\bt v + \alpha_3 \b u\th v+\alpha_4 v\th{\b u}=0. \label{BKP-2c-2}
\end{gather}
\end{subequations}
Both of them are 4D consistent by virtue of Theorem~\ref{T5}.

\section{Reductions}\label{sec-4}

\subsection{Nonlocal systems}\label{sec-4-1}

A few years ago, nonlocal integrable systems were introduced by Ablowitz and Musslimani~\cite{AM}. They studied the nonlocal nonlinear Schr\"{o}dinger equation
\[iq_t(x,t)+q_{xx}(x,t)+2 q^2(x,t)q^*(-x,t)=0,\]
where $i$ is the imaginary unit and $q^*$ denotes the complex conjugate of $q$. The equation is called nonlocal as it involves functions which depend on points $-x$, $x$ which are far apart. Most nonlocal integrable systems are continuous or semi-discrete. In this section we show that nonlocal fully discrete integrable equations can be constructed as reductions of 2-component ABS systems.

For a $2[0,1]$ ABS system, \eqref{2c1}, that is a system of the form
\begin{subequations}\label{2c1xy}
\begin{gather}
Q\big(u, \t u, \h v, \th v;p,q\big)=0,\label{2c1xy1} \\
Q\big(v, \t v, \h u, \th u;p,q\big)=0,\label{2c1xy2}
\end{gather}
\end{subequations}
where $u$ and $v$ are scalar functions of $(n,m)$. Introducing the relation
\begin{gather}\label{reduc-1}
v(n,m)=u(-n,m),
\end{gather}
\eqref{2c1xy1} reduces to
\begin{gather}\label{nonl-1}
 Q(u(n,m), u(n+1,m), u(-n,m+1), u(-n-1,m+1); p,q)=0.
\end{gather}
Equation (\ref{nonl-1}) is a nonlocal ABS equation. In nonlocal reduction, usually the coupled system does not collapse to one single equation, but solutions of the second equation can be obtained from the first one. Here, if $u(n,m)$ is a solution of (\ref{nonl-1}), then $v(n,m)$ given by~\eqref{reduc-1} is solution of the equation reduced from~(\ref{2c1xy2}). This is similar to the continuous case~\cite{AM}.

A $2[1,1]$ ABS system \eqref{2c2}, i.e.,
\begin{subequations}\label{2c2xy}
\begin{gather}
Q\big(u, \t v, \h v, \th u;p,q\big)=0, \label{2c2xy1}\\
Q\big(v, \t u, \h u, \th v;p,q\big)=0,\label{2c2xy2}
\end{gather}
\end{subequations}
admits the reduction
\begin{gather*}%\label{reduc-2}
v(n,m)=u(-n,-m),
\end{gather*}
which reduces \eqref{2c2xy1} to the nonlocal equation
\begin{gather}\label{nonl-2}
 Q(u(n,m), u(-n-1,-m), u(-n,-m-1), u(n+1,m+1); p,q)=0.
\end{gather}

As examples, the nonlocal H1 equation of type \eqref{nonl-1} reads
\begin{gather}\label{nlH1-1}
(u(n,m)-u(-n-1,m+1))(u(n+1,m)-u(-n,m+1))+q-p=0,
\end{gather}
and the nonlocal H1 equation of type \eqref{nonl-2} is
\begin{gather}\label{nlH1-2}
(u(n,m)-u(n+1,m+1))(u(-n-1,-m)-u(-n,-m-1))+q-p=0.
\end{gather}

We note that for any ABS equation which is invariant under the transformation $u\to -u$, its 2-component extension \eqref{2c1xy} allows nonlocal reduction by introducing
\begin{gather*}%\label{reduc-1e}
v(n,m)=\epsilon u(-n,m), \qquad \epsilon=\pm 1,
\end{gather*}
and its 2-component extension \eqref{2c2xy} allows nonlocal reduction
\begin{gather*}%\label{reduc-2e}
v(n,m)=\epsilon u(-n,-m), \qquad\epsilon=\pm 1.
\end{gather*}

Thus, besides \eqref{nlH1-1} and \eqref{nlH1-2}, the H1 equation also has nonlocal forms
\begin{gather*}%\label{nlH1-1e}
(u(n,m)+u(-n-1,m+1))(u(n+1,m)+u(-n,m+1))+q-p=0,
\end{gather*}
and
\begin{gather}\label{nlH1-2e}
(u(n,m)-u(n+1,m+1))(u(-n-1,-m)-u(-n,-m-1))-q+p=0.
\end{gather}

\subsection{Higher order equations from eliminations}\label{sec-4-2}
Multi-component extensions can be reduced to higher order lattice equations by elimination of field components. Examples of higher order lattice equations can be found in \cite{BS,DGNS,FN,GSY,KN,LMW,NPCQ,Sur}. The first instance of such an elimination procedure appeared in the theory of integrable discretisation of holomorphic and harmonic functions \cite{Duf,Mer}. More recently such techniques were applied in discrete integrable systems, e.g., in~\cite{KN}, where
multi-quadratic relations were obtained, and related to Yang-Baxter difference systems.

\subsubsection*{The discrete Burgers equation}
An example where the elimination can be done by a single substitution, is provided by the discrete Burgers equation. Eliminating the variable $v$ in the 2[0,1] extension (\ref{210B}) yields the four point
equation
\[
\frac{\tth u \tt u - q \tt u- \tth u (p-q)}{\tth u - p} \left(p-q+\h u-\frac{\tt u \tt {\dh u} - q \tt {\dh u}- \tt u (p-q)}{\tt u - p}\right)=p \h u-q \frac{\tt u \tt {\dh u} - q \tt {\dh u}- \tt u (p-q)}{\tt u - p}.
\]
Equations on similar but larger stencils are easily obtained from $N$-component extensions.

\subsubsection*{ABS equations}
For ABS-equations, the generic form of the function $Q$ is, cf.~\cite{Via},
\begin{gather}
Q\big(u,\t u,\h u,\th u\big) =k_1u\t u\h u\th u+k_2\big(u\t u\h u+u\h u\th u+u\t u\th u+\t u\h u\th u\big)+k_3\big(\t u\h u+u\th u\big) \nonumber\\
\hphantom{Q\big(u,\t u,\h u,\th u\big) =}{} +k_4\big(u\t u+\h u\th u\big)+k_5\big(u\h u+\t u\th u\big)
 +k_6\big(u+\t u+\h u+\th u\big)+k_7,\label{QQQQ}
\end{gather}
where the coefficients depend on the lattice parameters $p$, $q$.
Due the D4 symmetry property
\[
Q\big(u,\t u,\h u,\th u;p,q\big)=\pm Q\big(\h u,\th u,u,\t u;p,q\big),
\]
there exists a function $G$ such that
\begin{gather*}%\label{G}
\th u=G\big(\h u, u,\t u\big),\qquad \t u= G\big(u,\h u,\th u\big).
\end{gather*}

\subsubsection*{From $\boldsymbol{2[0,1]}$ ABS equations to scalar six-point equations}
For the $2[0,1]$ ABS equation \eqref{2c1xy}, we have
\begin{gather*}
\t v=G\big(v,\dh u,\dh{\t u}\big),\qquad \t v=G\big(v,\h u,\th u\big).
\end{gather*}
The equation $G\big(v,\dh u,\dh{\t u}\big) = G\big(v,\h u,\th u\big)$ is quadratic in $v$, so there exists a rational function $F$, linear in $\omega$,
such that
\begin{gather} \label{FT}
v=F\big(\dh u, \dh{\t u}, \h u,\th u; \omega\big),
\end{gather}
with
\begin{gather}\label{w2}
\omega^2=\varphi\big(\dh u, \dh{\t u}, \h u,\th u\big).
\end{gather}
Substituting \eqref{FT} into \eqref{2c1xy1} gives rise to
\begin{gather*}%\label{xeq0}
Q\big(\dh{\dt u}, \dh u, F\big(\dh{\dt u}, \dh u, \dt{\h u},\h u;\dt{\omega}\big),F\big(\dh u, \dh{\t u}, \h u,\th u;\omega\big)\big)=0,
\end{gather*}
which is multi-linear in the radicals $\omega$ and $\dt\omega$, i.e., of the form $m_{\omega,\dt\omega}:=c_1+c_2\omega+c_3\dt\omega+c_4\omega\dt\omega=0$.
Combining the different roots, using the identity
\[
m_{\omega,\dt\omega}m_{-\omega,\dt\omega}m_{\omega,-\dt\omega}m_{-\omega,-\dt\omega}
= \big( {c_1}^{2}{\omega}^{2}{{\dt\omega}}^{2}-{c_2}^{2}{\omega}^{2}
-{c_3}^{2}{{\dt\omega}}^{2}+{c_4}^{2} \big) ^{2}
- ( 2 c_1 c_4-2 c_2 c_3
 ) ^{2}{\omega}^{2}{{\dt\omega}}^{2}
\]
and substituting in the expressions for $\omega^2$, ${\dt \omega}^2$, (\ref{w2}),
one obtains a six-point equation of the form (see Fig.~\ref{F:4}(a))
\begin{gather} \label{xeq1}
H\big(\dh{\dt u}, \dt{\h u}, \dh{\t u}, \dh u,\h u,\th u\big)=0,
\end{gather}
which is multi-quadratic in each variable. Multi-quadratic CAC equations related to the ABS-equations have been studied in the literature, cf.~\cite{AN,KN}. The multi-quadratic equations in \cite{AN,KN} have been written in quadrilateral form. However, considering, e.g., d$Q1^\ast$ \cite[Table~5]{KN}, in the lpKdV variable $x$, which is related to $u$ by $u=\tilde{x}-x$, the equation lives on a six-point stencil.

For our $2[0,1]$H1 equation, we have
\[
\omega^2 =\big(\h u-\dh u\big)\big(\th u-\t{\dh u}\big)\big(\big(\h u-\dh u\big)\big(\th u-\t{\dh u}\big)+4q-4p\big),
\]
and equation \eqref{xeq1} yields the six-point equation
\begin{gather}
 \big(\th u-\h{\dt u}\big)\big(\th u-\dh{\dt u}\big)
\big(\t{\dh u}-\dh{\dt u}\big)\big(\t{\dh u}-\h{\dt u}\big)\big(\h u-\dh u\big)^2+2(p-q)\big(\th u-\h{\dt u}\big)\big(\t{\dh u}-\dh{\dt u}\big)\big(\h u-\dh u\big)\big(\th u-\t{\dh u}-\dt{\dh u}+\h{\dt u}\big)\nonumber\\
\qquad{} +(p-q)^2\big(\th u-\t{\dh u}+\dt{\dh u}-\h{\dt u}\big)^2=0.\label{201H1}
\end{gather}
To our knowledge this is a new equation. It can be written as a quadrilateral system by introducing variables $x=\tilde{u}-u$, $y=\hat{\hat{u}}-u$.

Note that due to the symmetric positions for $u$ and $v$ in~\eqref{2c1xy}, variable $v$ satisfies \eqref{xeq1} as well, and therefore system~\eqref{2c1xy} can serve as either a~Lax pair or an auto-BT for~\eqref{xeq1}.

\subsubsection*{From $\boldsymbol{2[1,1]}$ ABS equations to scalar five-point equations}
For the $2[1,1]$ ABS equation, the second component \eqref{2c2xy2} can be written variously as
\begin{gather} \label{thy}
\t v=G\big(u,\dh v,\t{\dh u}\big),\qquad \h v=G\big(\h{\dt u},\dt v,u\big),
\end{gather}
and the first component \eqref{2c2xy1} is equivalent to
\begin{gather} \label{dhy}
\dh v=G\big(\dt{\dh u},\dt v,u\big).
\end{gather}
Substituting the formulas \eqref{thy} and \eqref{dhy} into \eqref{2c2xy1} gives
\begin{gather} \label{5p}
Q\big(u,G\big(u,G(\dt{\dh u},\dt v,u),\t{\dh u}\big),G\big(\h{\dt u},\dt v,u\big),\th u\big)=0.
\end{gather}
Miraculously, when $Q$ has the form \eqref{QQQQ} the equation~\eqref{5p} factorises. One factor is biquadratic in $u$, $\dt v$ and the other is
\begin{gather}
H\big(\dt{\dh u}, \h{\dt u},\t{\dh u}, u, \th u\big) =
 \big(l_1u^2\!+l_2u+l_3\big)\big(\dt{\dh u}\t{\dh u}\h{\dt u}-\dt{\dh u}\t{\dh u}\th u-\dt{\dh u}\th u\h{\dt u}+\th u\h{\dt u}\t{\dh u}\big)
 +\big(l_2u^2\!+l_4u+l_5\big)\big(\h{\dt u}\t{\dh u}-\th u\dt{\dh u}\big)\nonumber\\
\hphantom{H\big(\dt{\dh u}, \h{\dt u},\t{\dh u}, u, \th u\big) =}{}
+\big(l_3u^2+l_5u+l_6\big)\big(\h{\dt u}+\t{\dh u}-\th u-\dt{\dh u}\big)=0\label{xeq2}
\end{gather}
with parameters defined by
\begin{gather*}
 l_1=k_1k_3-k_2^2,\qquad l_2=-k_5k_2-k_2k_4+k_1k_6+k_3k_2,\qquad l_3= -k_4k_5+k_6k_2\\
 l_4=-k_4^2+k_3^2-k_5^2+k_1k_7,\qquad l_5= k_7k_2-k_6k_4-k_5k_6+k_6k_3,\qquad l_6 =k_7k_3-k_6^2,
\end{gather*}
quadratic in $u$ and multi-linear in $\dt{\dh u}$, $\t{\dh u}$, $\h{\dt u}$, $\th u$, see Fig.~\ref{F:4}(b).
We note that~\eqref{2c2xy} provides a Lax pair as well as an auto-BT for~\eqref{xeq2}.

\begin{figure}[h]\centering
\setlength{\unitlength}{0.06em}
\begin{picture}(200,170)(0,-20)
\drawline(0,20)(140,140)
\drawline(0,140)(140,140)
\drawline(0,20)(140,20)
\drawline(140,20)(0,140)
\put(-15,10){\makebox(0,0)[lb]{$ \dt{\dh u}$}}
\put(0,20){\circle*{5}}
\put(70,6){\makebox(0,0)[lb]{$ \dh u$}}
\put(70,20){\circle*{5}}
\put(-15,130){\makebox(0,0)[lb]{$ \h{\dt u}$}}
\put(0,140){\circle*{5}}
\put(70,120){\makebox(0,0)[lb]{$ \h u$}}
\put(70,140){\circle*{5}}
\put(145,130){\makebox(0,0)[lb]{$\th u$}}
\put(140,140){\circle*{5}}
\put(145,10){\makebox(0,0)[lb]{$\t{\dh u}$}}
\put(140,20){\circle*{5}}
 \put(-30,-20){\small (a) Stencil of equaton \eqref{xeq1}}
\end{picture}
\hspace{2cm}
\begin{picture}(200,170)(0,-20)
\drawline(0,20)(0,160)
\drawline(0,160)(140,20)
\drawline(0,20)(140,160)
\drawline(140,20)(140,160)
\put(-15,10){\makebox(0,0)[lb]{$ \dt{\dh u}$}}
\put(0,20){\circle*{5}}
\put(145,10){\makebox(0,0)[lb]{$ \t{\dh u}$}}
\put(140,20){\circle*{5}}
\put(-15,150){\makebox(0,0)[lb]{$ \h{\dt u}$}}
\put(0,160){\circle*{5}}
\put(145,150){\makebox(0,0)[lb]{$\th u$}}
\put(140,160){\circle*{5}}
\put(65,75){\makebox(0,0)[lb]{$u$}}
\put(70,90){\circle*{5}}
 \put(-30,-20){\small (b) Stencil of equation \eqref{xeq2}}
\end{picture}
\caption{Equations on five- and six-point stencils are obtained from 2-component ABS equations. \label{F:4}}
\end{figure}

The 5-point equation \eqref{xeq2} is not new; it can also be obtained by elimination of the single shifts from 4 copies of (\ref{QQQQ}) on a $2\times2$ stencil, and it coincides with the discrete Toda-type equations classified by Adler in \cite{Adl1,Adl2}. This can be seen as follows. The ABS equations admit a~three-leg form \cite[Section 5]{ABS03}, either an additive form,
\begin{gather*} %\label{form1}
\psi(u,\t u;p)-\psi(u,\h u;q)=\phi\big(u,\th u;p, q\big),
\end{gather*}
or an multiplicative form,
\begin{gather*}%\label{form2}
\Psi(u,\t u;p)/\Psi(u,\h u;q)=\Phi\big(u,\th u;p, q\big),
\end{gather*}
with certain functions $\psi$, $\phi$, $\Psi$ and $\Phi$. In the additive case, by means of $D4$ symmetry, the $2[1,1]$ equations we have used to derive \eqref{xeq2} can be expressed in terms of the three-leg form as
\begin{gather*}
 Q\big(u, \t v, \h v, \th u;p,q\big)=0\Leftrightarrow \psi(u,\t v;p)-\psi(u,\h v;q)=\phi\big(u,\th u;p, q\big),\\
 Q(\dt v, u, \h{\dt u}, \h v;p,q)=0 \Leftrightarrow Q( u,\dt v, \h v, \h{\dt u};p,q)=0
\Leftrightarrow \psi(u,\dt v;p)-\psi(u,\h v;q)=\phi(u,\h{\dt u};p, q),\\
 Q(\dh v, \t{\dh u}, u, \t v;p,q)=0 \Leftrightarrow Q(u, \t v, \dh v, \t{\dh u};p,q)=0
\Leftrightarrow \psi(u,\t v;p)-\psi(u,\dh v;q)=\phi(u, \t{\dh u};p, q),\\
 Q(\dt{\dh u}, \dh v, \dt v, u;p,q)=0 \Leftrightarrow Q(u,\dt v,\dh v,\dt{\dh u};p,q)=0
\Leftrightarrow \psi(u,\dt v;p)-\psi(u,\dh v;q)=\phi(u,\dt{\dh u};p, q).
\end{gather*}
Eliminating $\psi$ we arrive at a five-point lattice equation
\begin{gather*} %\label{form1s}
\phi\big(u,\th u;p, q\big) +\phi(u,\h{\dt u};p, q)=\phi(u, \t{\dh u};p, q)+\phi(u,\dt{\dh u};p, q),
\end{gather*}
which is called a discrete Toda-type equation in \cite{Adl1,Adl2}. In the multiplicative case, we have
\begin{gather*} %\label{form-m}
\Phi\big(u,\th u;p, q\big) \Phi(u,\dt{\dh u};p, q)= \Phi(u,\h{\dt u};p, q) \Phi(u, \t{\dh u};p, q),
\end{gather*}
which is again a discrete Toda-type equation after a transformation $\Phi={\rm e}^{\phi}$.

Explicitly, the 5-point equation derived from $2[1,1]$H1 or $2[1,1]\text{Q1}^0$ can be written as
\begin{gather*}%\label{Todaa}
\frac{1}{u-\dt{\dh u}}+\frac{1}{u-\th u}=\frac{1}{u-\t{\dh u}}+\frac{1}{u-\h{\dt u}}.
\end{gather*}
This equation can also be found by elimination of the single shifts in the 7-point equation \cite[equation~(3)]{NPC}. The systems $2[1,1]$H2, $2[1,1]\text{Q1}^1$ and $2[1,1]\text{A1}^1$ give rise to
\begin{gather*}%\label{Todac}
\frac{\big(u-\th u+k\big)(u-\dt{\dh u}+k)}{\big(u-\th u-k\big)(u-\dt{\dh u}-k)}=
\frac{(u-\h{\dt u}+k)(u-\t{\dh u}+k)}{(u-\h{\dt u}-k)(u-\t{\dh u}-k)},\qquad k=p-q,
\end{gather*}
and $2[1,1]\text{H3}^\delta$, $2[1,1]$A2 yield
\begin{gather*}%\label{Todad}
\frac{\big(ku-\th u\big)(ku-\dt{\dh u})}{\big(u-k\th u\big)(u-k\dt{\dh u})}=
\frac{(ku-\h{\dt u})(ku-\t{\dh u})}{(u-k\h{\dt u})(u-k\t{\dh u})},\qquad k=p/q.
\end{gather*}
In \cite{KNPT20} certain $a$- and $m$-bond systems lead to double H-type vertex systems, cf.\ \cite[Proposition~6.1]{KNPT20}. As mentioned in \cite[Remark~6.2]{KNPT20}, two of these are two-component extensions of the type discussed here, the third is a mixed 2-component system, H1$\times$H2, and they provide B\"acklund transformations for the above 5-point schemes.

The other $2[1,1]$ABS equations relate to the following 5-point equations:
\begin{align*}
2[1,1]\text{Q2}\colon \quad & \dt{\dh u}\t{\dh u}\h{\dt u}-\dt{\dh u}\t{\dh u}\th u-\dt{\dh u}\th u\h{\dt u}
+\th u\h{\dt u}\t{\dh u}-2\big(u+k^2\big)\big(\h{\dt u}\t{\dh u}-\th u\dt{\dh u}\big)\nonumber\\
&{}+\big(u-k^2\big)^2\big(\h{\dt u}+\t{\dh u}-\th u-\dt{\dh u}\big)=0,\qquad k=p-q,\\ %\label{211Q2}\\
2[1,1]\text{Q3}\colon \quad &\dt{\dh u}\t{\dh u}\h{\dt u}-\dt{\dh u}\t{\dh u}\th u-\dt{\dh u}\th u\h{\dt u}
+\th u\h{\dt u}\t{\dh u}-\big(k+k^{-1}\big)u\big(\h{\dt u}\t{\dh u}-\th u\dt{\dh u}\big)\nonumber\\
&{} +\left(u^2+\frac{\delta^2\big(k-k^{-1}\big)^2}{4}\right)\big(\h{\dt u}+\t{\dh u}-\th u-\dt{\dh u}\big)=0,\qquad k=p/q,\\ %\label{211Q3}\\
2[1,1]\text{Q4}\colon \quad &\big(1-k^2u^2\big)\big(\dt{\dh u}\t{\dh u}\h{\dt u}-\dt{\dh u}\t{\dh u}\th u
-\dt{\dh u}\th u\h{\dt u}+\th u\h{\dt u}\t{\dh u}\big)-2\big(1-k^2\big)u\big(\h{\dt u}\t{\dh u}-\th u\dt{\dh u}\big)\nonumber\\
&{} +\big(u^2-k^2\big)\big(\h{\dt u}+\t{\dh u}-\th u-\dt{\dh u}\big)=0,\qquad k=\mathrm{sn}(p-q).%\label{211Q4}
\end{align*}

\subsubsection*{Higher order scalar equations from $\boldsymbol{2[1,1,1]}$ AKP and BKP}
Here we show that higher order equations can also be obtained from multi-component extensions of 3D lattice equations.

One can eliminate $u$ from the $2[1,1,1]$ AKP equation \eqref{AKP-2c} as follows.
Denote the left hand side from \eqref{AKP-2c-1} by $E$ and the left hand side from \eqref{AKP-2c-2} by $F$. First we solve
$\h {\dt u}$ from $\dt E=0$, $\h {\db u}$ from $\db E=0$,
$\t {\dh u}$ from $\db{\dt F}=0$ and $\db {\dh u}$ from $\dh{\dt F}=0$.
Then we substitute $\h {\dt u}, \h {\db u}$ into $\db{\dt F}/(\dt v \db v\h v)$
and~$\t {\dh u}$,~$\b {\dh u}$ into $\dh E/(\dh v \t v\b v)$. The difference of the results gives rise to the 12-point equation
\begin{gather}\label{dAKP-12}
a_1^2\left[\frac{\dhb{\t v}}{\dh v\db v\t v}-\frac{\bh{\dt v}}{\dt v\h v\b v}\right]
+a_2^2\left[\frac{\dtb{\h v}}{\dt v\db v\h v}-\frac{\bt{\dh v}}{\dh v\t v\b v}\right]
+a_3^2\left[\frac{\dth{\b v}}{\dt v\dh v\b v}-\frac{\th{\db v}}{\db v\t v\h v}\right] =0.
\end{gather}
Similar to the 2D case, the coupled system \eqref{AKP-2c} can be considered as either a BT or a Lax pair of~\eqref{dAKP-12}.

From the 2-component BKP equation \eqref{BKP-2c}, using a similar elimination scheme, one obtains the 14-point King--Schief equation
\begin{gather}\label{dBKP-14}
a_1^2\left[\frac{\dhb{\t v}}{\dh v\db v\t v}-\frac{\bh{\dt v}}{\dt v\h v\b v}\right]
+a_2^2\left[\frac{\dtb{\h v}}{\dt v\db v\h v}-\frac{\bt{\dh v}}{\dh v\t v\b v}\right]
+a_3^2\left[\frac{\dth{\b v}}{\dt v\dh v\b v}-\frac{\th{\db v}}{\db v\t v\h v}\right]
+a_4^2\left[\frac{\th{\b v}}{\t v\h v\b v}-\frac{\db{\dth v}}{\db v\dt v\dh v}\right]
=0,
\end{gather}
which arose in the study of nondegenerate Cox lattices \cite{King-Schief-2015}. Its BT/Lax pair is provided by~\eqref{BKP-2c}, with $u$ acting as an eigenvalue function, cf.\ \cite[equation~(30)]{King-Schief-2015}, and the equation degenerates to equation \eqref{dAKP-12} when $a_4=0$. The equations~\eqref{dAKP-12} and \eqref{dBKP-14} are satisfied by solutions of the AKP equation and the BKP equation respectively.

The analysis in \cite[Section~5]{King-Schief-2015} reveals that Cox--Menelaus lattices are intimately related to the AKP equation. We expect that \eqref{dAKP-12} will play a similar role in that context as (\ref{dBKP-14}) plays in the context of nondegenerate Cox lattices. We further note that equation (\ref{dBKP-14}) relates, by a simple coordinate transformation \cite{KZQ}, to an equation that appeared in a completely different setting, namely through the notion of duality, employing conservation laws of the lattice AKP equation~\cite{Kamp-JPA-2018}.

\subsection[$N$-component $\mathfrak{q}$-Painlev\'e III equation]{$\boldsymbol{N}$-component $\boldsymbol{\mathfrak{q}}$-Painlev\'e III equation}\label{sec-4-3}

The results in Section~\ref{sec-2} remain true for non-autonomous multi-dimensionally consistent systems,
extending spacing parameters $p \to p(n)$, $q\to q(m)$ and $r\to r(l)$.
There are close relations, cf.~\cite{JNS-JIS-2016, NN}, between non-autonomous ABS lattice equations and
discrete Painlev\'e equations exploing the affine Weyl group.
A particular example is provided by a non-autonomous version of
the lpmKdV equation~(\ref{H30}) which
can be reduced to a $\mathfrak{q}$-Painlev\'e III equation, by
performing a periodic reduction. We use this example to illustrate
that such a link extends to the multi-component case.

For the non-autonomous lpmKdV equation
\[
Q\big(u,\t u,\h u,\th u\big)=p\big(u\t u-\h u\th u\big)-q\big(u\h u-\t u\th u\big)=0,\qquad p=p_0\mathfrak{q}^n,\qquad q=q_0\mathfrak{q}^m,
\]
we use the bottom and front equations on its multi-component consistent cube, i.e.,
\begin{subequations}\label{mkdv-t}
\begin{gather}
p\big( u T_{a}\t u-(T_{b}\h u)\big(T_{a+b}\th u\big)\big)-q\big(u T_{b}\h u-(T_{a}\t u)\big(T_{a+b}\th u\big)\big)=0, \label{mkdv-ta}\\
p\big( u T_{a}\t u -(T_{c}\b u)\big(T_{a+c}\tb u\big)\big)-r\big(uT_{c}{\b u}-(T_{a}\t u)\big(T_{a+c}\tb u\big)\big)=0. \label{mkdv-tb}
\end{gather}
\end{subequations}
Imposing the periodic reduction (cf.~\cite{JNS-JIS-2016})
\begin{gather}\label{red}
\thb u=u,\qquad a+b+c=0
\end{gather}
on \eqref{mkdv-t}, and replacing $p$, $q$ by $p\mathfrak{q}$, $q\mathfrak{q}$, and the first equation \eqref{mkdv-ta} is unchanged but we rewrite it as
\begin{subequations}\label{mkdv-tt}
\begin{gather}\label{mkdv-tta}
T_{a}\t u\big( p u + q T_{a+b}\th u\big)= T_{b}\h u\big(q u+ pT_{a+b}\th u\big)
\end{gather}
for convenience.
For \eqref{mkdv-tb}, after a tilde/hat-shift and making use of the reduction \eqref{red}, we have
\begin{gather}\label{mkdv-ttb}
T_{a}\tth u \big(p\mathfrak{q}\th u+r T_{-b}\t u\big)=T_{-a-b}u \big(r\th u+p\mathfrak{q}T_{-b}\t u\big).
\end{gather}
\end{subequations}
Then, introducing diagonal matrices $f$ and $g$ by
\[f=(T_{a}\t u)\big(T_{a+b}\th u\big)^{-1},\qquad
g=\big(T_{a+b}\th u\big) u^{-1},\qquad
 t=\frac{\mathfrak{q}p}{r},\qquad k=\frac{r}{q\mathfrak{q}},\]
from \eqref{mkdv-tta} and \eqref{mkdv-ttb} we find
\begin{gather*}
T_{-a}\undertilde f =\frac{1+ktg}{fg(kt+g)},\qquad
T_{a}\t g =\frac{1+tf}{fg(t+f)},%\label{P-III}
\end{gather*}
which is an $N$-component $\mathfrak{q}$-Painlev\'e III equation.

Taking $N=3$ and $a=1$, this yield the 3-component $\mathfrak{q}$-Painlev\'e III system
\begin{alignat*}{3}
& \undertilde f_3=\frac{1+ktg_1}{f_1g_1(kt+g_1)},\qquad && \t g_2=\frac{1+tf_1}{f_1g_1(t+f_1)},& \\
& \undertilde f_1=\frac{1+ktg_2}{f_2g_2(kt+g_2)},\qquad && \t g_3=\frac{1+tf_2}{f_2g_2(t+f_2)},& \\
& \undertilde f_2=\frac{1+ktg_3}{f_3g_3(kt+g_3)}, \qquad && \t g_1=\frac{1+tf_3}{f_3g_3(t+f_3)}.&
\end{alignat*}

\section{Exact solutions}\label{sec-5}

\subsection{Solutions with jumping property}\label{sec-5-1}
Using $N$ solutions to the scalar equation, one can construct a solution for the
$N$-component equation \eqref{QQ}. For $j=1,2,\dots,N$, let $w_j(n,m)$ be a solution to the scalar equation \eqref{Q0}, i.e.,
\begin{gather}\label{Q01}
Q\big(w_j,\t w_j,\h w_j,\th w_j; p,q\big)=0.
\end{gather}
If we set
\begin{gather}\label{sol}
u_k(n,m)=w_{k-an-bm}(n,m),
\end{gather}
where the sub-index is taken modulo $N$, the system of equations \eqref{QQ} comprises $N$ copies of the scalar equation.
Thus, \eqref{sol} provides a solution to \eqref{QQ}. If \eqref{QQ} is non-decoupled, i.e., when $\operatorname{gcd}(a,b,N)=1$,
then $u_k$ will run over $\{w_j\colon j=1,2,\dots N\}$ by virtue of Lemma~\ref{L2} in Appendix~\ref{adec}.
The pattern of $u_k$ is depicted in Fig.~\ref{F:5}.

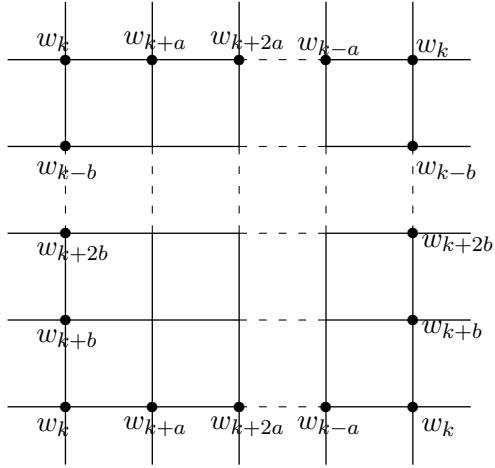
\begin{figure}[ht]\centering
\setlength{\unitlength}{0.10em}
\hspace{20mm} \begin{picture}(200,170)(15,-5)
\drawline(0,20)(80,20)
\drawline(0,50)(80,50)
\drawline(0,80)(80,80)
\drawline(0,110)(80,110)
\drawline(0,140)(80,140)
\dashline{3.000}(110,20)(80,20)
\dashline{3.000}(110,50)(80,50)
\dashline{3.000}(110,80)(80,80)
\dashline{3.000}(110,110)(80,110)
\dashline{3.000}(110,140)(80,140)
\drawline(110,20)(160,20)
\drawline(110,50)(160,50)
\drawline(110,80)(160,80)
\drawline(110,110)(160,110)
\drawline(110,140)(160,140)

\drawline(20,0)(20,80)
\drawline(50,0)(50,80)
\drawline(80,0)(80,80)
\drawline(110,0)(110,80)
\drawline(140,0)(140,80)
\dashline{3.000}(20,110)(20,80)
\dashline{3.000}(50,110)(50,80)
\dashline{3.000}(80,110)(80,80)
\dashline{3.000}(110,110)(110,80)
\dashline{3.000}(140,110)(140,80)
\drawline(20,110)(20,160)
\drawline(50,110)(50,160)
\drawline(80,110)(80,160)
\drawline(110,110)(110,160)
\drawline(140,110)(140,160)

\put(10,10){\makebox(0,0)[lb]{$w_k$}}
\put(20,20){\circle*{4}}
\put(40,10){\makebox(0,0)[lb]{$w_{k+a}$}}
\put(50,20){\circle*{4}}
\put(70,10){\makebox(0,0)[lb]{$w_{k+2a}$}}
\put(80,20){\circle*{4}}
\put(100,10){\makebox(0,0)[lb]{$w_{k-a}$}}
\put(110,20){\circle*{4}}
\put(143,10){\makebox(0,0)[lb]{$w_k$}}
\put(140,20){\circle*{4}}

\put(10,143){\makebox(0,0)[lb]{$w_k$}}
\put(20,140){\circle*{4}}
\put(40,143){\makebox(0,0)[lb]{$w_{k+a}$}}
\put(50,140){\circle*{4}}
\put(70,143){\makebox(0,0)[lb]{$w_{k+2a}$}}
\put(80,140){\circle*{4}}
\put(100,141){\makebox(0,0)[lb]{$w_{k-a}$}}
\put(110,140){\circle*{4}}
\put(141,141){\makebox(0,0)[lb]{$w_k$}}
\put(140,140){\circle*{4}}

\put(10,40){\makebox(0,0)[lb]{$w_{k+b}$}}
\put(20,50){\circle*{4}}
\put(10,70){\makebox(0,0)[lb]{$w_{k+2b}$}}
\put(20,80){\circle*{4}}
\put(10,98){\makebox(0,0)[lb]{$w_{k-b}$}}
\put(20,110){\circle*{4}}

\put(143,43){\makebox(0,0)[lb]{$w_{k+b}$}}
\put(140,50){\circle*{4}}
\put(143,73){\makebox(0,0)[lb]{$w_{k+2b}$}}
\put(140,80){\circle*{4}}
\put(141,98){\makebox(0,0)[lb]{$w_{k-b}$}}
\put(140,110){\circle*{4}}
\end{picture}
\caption{$u_k$ defined by \eqref{sol} on $(n,m)$ lattice.}\label{F:5}
\end{figure}

For the $2[1,1]$ ABS equation \eqref{2c2}, according to \eqref{sol} its solution can be given by
%\bse\label{2cABS:sol}
\begin{gather*}
u_1=
\begin{cases}
w_1, & n+m\equiv 0\ (\text{mod} \ 2),\\
w_2, & n+m\equiv 1,
\end{cases} \qquad u_2=
\begin{cases}
w_2, & n+m\equiv 0,\\
w_1, & n+m\equiv 1,
\end{cases}
\end{gather*}
where each $w_i$ satisfies scalar equation \eqref{Q01}. This coincides with the result in~\cite{FZ14}.
For the $3[1,1]$ ABS equation~\eqref{QQ}, solutions can be presented by
\begin{equation}\label{sol3c}
u_k=
\left\{
\begin{array}{ll}
w_k, & n+m\equiv 0\ (\text{mod} \ 3),\\
w_{k-1}, & n+m\equiv 1,\\
w_{k-2}, & n+m\equiv 2,
\end{array}
\right.
\end{equation}
provided each $w_i$ solves the scalar equation \eqref{Q01}. Note that \eqref{sol} has the so-called jumping property (cf.~\cite{FZ14}) and
for~\eqref{sol3c} this property is illustrated by Fig.~\ref{F:6}.

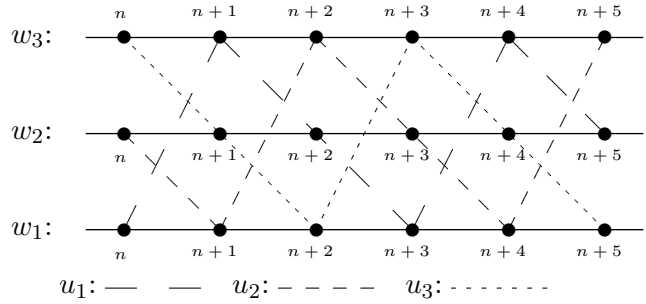
\begin{figure}[h]\centering
\setlength{\unitlength}{0.0005in}
\vspace{5mm}

\hspace{-35mm}\begin{picture}(3482,2813)(300,-10)
%\put(-800,1510){\makebox(0,0)[lb]{$(a)$}}
\put(100,2608){\makebox(0,0)[lb]{$w_3$:}}
\put(1275,2708){\circle*{150}}
\put(1175,2908){\makebox(0,0)[lb]{\tiny{$n$}}}
\put(2275,2708){\circle*{150}}
\put(1975,2908){\makebox(0,0)[lb]{\tiny{$n+1$}}}
\put(3275,2708){\circle*{150}}
\put(2975,2908){\makebox(0,0)[lb]{\tiny{$n+2$}}}
\put(4275,2708){\circle*{150}}
\put(3975,2908){\makebox(0,0)[lb]{\tiny{$n+3$}}}
\put(5275,2708){\circle*{150}}
\put(4975,2908){\makebox(0,0)[lb]{\tiny{$n+4$}}}
\put(6275,2708){\circle*{150}}
\put(5975,2908){\makebox(0,0)[lb]{\tiny{$n+5$}}}
\drawline(875,2708)(6675,2708)
\drawline(875,1708)(6675,1708)
\drawline(875,708)(6675,708)
\put(100,1608){\makebox(0,0)[lb]{$w_2$:}}
\put(1275,1708){\circle*{150}}
\put(1175,1408){\makebox(0,0)[lb]{\tiny{$n$}}}
\put(2275,1708){\circle*{150}}
\put(1975,1408){\makebox(0,0)[lb]{\tiny{$n+1$}}}
\put(3275,1708){\circle*{150}}
\put(2975,1408){\makebox(0,0)[lb]{\tiny{$n+2$}}}
\put(4275,1708){\circle*{150}}
\put(3975,1408){\makebox(0,0)[lb]{\tiny{$n+3$}}}
\put(5275,1708){\circle*{150}}
\put(4975,1408){\makebox(0,0)[lb]{\tiny{$n+4$}}}
\put(6275,1708){\circle*{150}}
\put(5975,1408){\makebox(0,0)[lb]{\tiny{$n+5$}}}

\put(100,608){\makebox(0,0)[lb]{$w_1$:}}
\put(1275,708){\circle*{150}}
\put(1175,408){\makebox(0,0)[lb]{\tiny{$n$}}}
\put(2275,708){\circle*{150}}
\put(1975,408){\makebox(0,0)[lb]{\tiny{$n+1$}}}
\put(3275,708){\circle*{150}}
\put(2975,408){\makebox(0,0)[lb]{\tiny{$n+2$}}}
\put(4275,708){\circle*{150}}
\put(3975,408){\makebox(0,0)[lb]{\tiny{$n+3$}}}
\put(5275,708){\circle*{150}}
\put(4975,408){\makebox(0,0)[lb]{\tiny{$n+4$}}}
\put(6275,708){\circle*{150}}
\put(5975,408){\makebox(0,0)[lb]{\tiny{$n+5$}}}
\dashline{50.000}(1275,2708)(3275,708)
\dashline{50.000}(4275,2708)(3275,708)
\dashline{50.000}(4275,2708)(6275,708)

\dashline{150.000}(1275,1708)(2275,708)
\dashline{150.000}(3275,2708)(2275,708)
\dashline{150.000}(3275,2708)(5275,708)
\dashline{150.000}(6275,2708)(5275,708)

\dashline{300.000}(1275,708)(2275,2708)
\dashline{300.000}(4275,708)(2275,2708)
\dashline{300.000}(4275,708)(5275,2708)
\dashline{300.000}(6275,1708)(5275,2708)

\put(600,0){\makebox(0,0)[lb]{$u_1$:}}
\dashline{300.000}(1075,100)(2075,100)
\put(2400,0){\makebox(0,0)[lb]{$u_2$:}}
\dashline{150.000}(2875,100)(3875,100)
\put(4200,0){\makebox(0,0)[lb]{$u_3$:}}
\dashline{60.000}(4675,100)(5675,100)
\end{picture}
\caption{Jumping property of $u_i$ in \eqref{sol3c} in the tilde-direction.}\label{F:6}
\end{figure}

Solutions for multi-component extensions of 3D equations, given in Section~\ref{sec-6}, can be given in a similar fashion as for 2D equations. If $\{w_j\}$ are $N$ solutions of the 3D scalar equation \eqref{Q3D},
then
\begin{gather*}%\label{sol-3D}
u_k(n,m,l)=w_{k-an-bm-cl}(n,m,l),
\end{gather*}
where the sub-index is taken modulo $N$, provides a solution of the $N[a,b,c]$ extension \eqref{QQ3D}.

\subsubsection*{Bilinear equations}\label{sec-5-2}

Many equations in the ABS list have been bilinearized \cite{HZ09}.
If a scalar ABS equation \eqref{Q0} has a bilinear form\footnote{Some equations need more than two functions to get bilinear forms.
Here we just employ \eqref{H-fg} as a generic form.}
\begin{gather} \label{H-fg}
H\big(\mathfrak{f},\mathfrak{g},\t{\mathfrak{f}}, \h{\mathfrak{f}},\h{\mathfrak{g}},\h{\mathfrak{g}},\th{\mathfrak{f}},\th{\mathfrak{g}}\big)=0,
\end{gather}
with transformation $u=F(\mathfrak{f},\mathfrak{g})$,
for example, H1 equation \eqref{H1-1c} has bilinear form
\begin{gather*}
\h g\t f-\t g \h f +(\alpha-\beta)\big(\t f \h f -f\th f \big)=0,\qquad
g \th f -\th g f +(\alpha+\beta)\big(f \th f -\t f \h f \big)=0,
\end{gather*}
where $p=-\alpha^2$, $q=-\beta^2$,
then for the $N[a,b]$ system \eqref{QQ}, its bilinear form can be given by
\begin{gather}\label{H-fg-Nc}
H\big(f,g,T_a\t f, T_a\t g, T_b\h f, T_b\h g, T_{a+b}\th f, T_{a+b}\th g\big)=0,
\end{gather}
through the transformation $u=F(f,g)$ where $f$, $g$ are diagonal forms in~\eqref{fg}.

With respect to solutions, suppose $(\mathfrak{f}_j,\mathfrak{g}_j)$ are any arbitrary solutions of \eqref{H-fg}.
Using them we define
\begin{gather*}
f_k(n,m)=\mathfrak{f}_{k-an-bm}(n,m),\qquad g_k(n,m)=\mathfrak{g}_{k-an-bm}(n,m).
\end{gather*}
Then, $(f,g)$ composed by such components will be a solution to~\eqref{H-fg-Nc}.

As an example, the $2[0,1]$ H1 equation \eqref{2cH11} has a bilinear form
\begin{gather*}
\h g_2\t f_1-\t g_1\h f_2+(\alpha-\beta)\big(\t f_1\h f_2-f_1\th f_2\big)=0,\\
\h g_1\t f_2-\t g_2\h f_1+(\alpha-\beta)\big(\t f_2\h f_1-f_1\th f_1\big)=0,\\
g_1\th f_2-\th g_2 f_1+(\alpha+\beta)\big(f_1\th f_2-\t f_1\h f_2\big)=0,\\
g_2\th f_1-\th g_1 f_2+(\alpha+\beta)\big(f_2\th f_1-\t f_2\h f_1\big)=0
\end{gather*}
with transformation $u_i=\alpha n+\beta m+r-{g_i}/{f_i}$,
and (with $m$ taken modulo~2)
\begin{gather*}
f_1=
\begin{cases}
\mathfrak{f}_1, & m\equiv 0,\\
\mathfrak{f}_2, & m\equiv 1,
\end{cases}\qquad\!
 f_2=
\begin{cases}
\mathfrak{f}_2, & m\equiv 0,\\
\mathfrak{f}_1, & m\equiv 1,
\end{cases} \qquad\!
 g_1=
\begin{cases}
\mathfrak{g}_1, & m \equiv 0,\\
\mathfrak{g}_2, & m \equiv 1,
\end{cases} \qquad \!
 g_2=
\begin{cases}
\mathfrak{g}_2, & m \equiv 0,\\
\mathfrak{g}_1, & m \equiv 1,
\end{cases}\!
\end{gather*}
where $(\mathfrak{f}_i,\mathfrak{g}_i)$ are solutions of~\eqref{H-fg}.

\subsection{Nonlocal case}\label{sec-5-3}

For some nonlocal ABS equations, if the scalar equation \eqref{Q0} admits an odd or even solution, i.e.,
\begin{gather*}
u(n,m)=\epsilon u(-n,-m),\qquad \epsilon=\pm 1,
\end{gather*}
then, such solutions may be used to construct a solution to the nonlocal equation.

As an example, let us look at the nonlocal H1 equation \eqref{nlH1-2e}. The local H1 equation \eqref{H1-1c}
has rational solution \cite{ZZ-SIGMA-2017}
\begin{gather}\label{H1-trans}
u=\frac{n}{\mu}+\frac{m}{\nu}-\frac{g_{[N]}}{f_{[N]}},
\end{gather}
where $p=-1/\mu^2$, $q=-1/\nu^2$, $f_{[N]}$ and $g_{[N]}$ are Casoratians \cite{ZZ-SIGMA-2017}
\begin{gather*}%\label{f-g}
 f_{[N]} =\big|\h{N-1}\big| =|\alpha(n,m,0),\alpha(n,m,1),\dots,\alpha(n,m,N-1)|, \\
 g_{[N]}=\big|\h{N-2},N\big| -Nf_{[N]}.
\end{gather*}
The Casoratian vector is
\[\alpha(n,m,l)=(\alpha_0, \alpha_1, \dots, \alpha_{N-1})^{\rm T},\qquad \alpha_j=\frac{1}{(2j+1)!}\partial^{2j+1}_{s_i}\psi_i|_{s_i=0},\]
with
\begin{equation*}
 \psi_i(n,m,l)= \psi_i^{+}(n,m,l) + \psi_i^{-}(n,m,l),\qquad
 \psi_i^{\pm}(n,m,l)= (1\pm s_i)^{l}(1\pm \mu s_i)^n(1\pm \nu s_i)^m.
\end{equation*}
$\psi_i^{\pm}$ has the form
\[\psi_i^{\pm}(n,m,l)=\pm \frac{1}{2}\sum^{\infty}_{h=0}\alpha^{\pm}_h s_i^h
=\pm \frac{1}{2}\exp \left[ -\sum^{\infty}_{j=1}\frac{(\mp s_i)^{j}}{j}\c{x}_j \right],\]
where
\begin{gather*}%\label{x-i}
 \c{x}_j=x_j+l,\qquad x_j=\mu^jn+\nu^jm,\qquad j\in \mathbb{Z}.
\end{gather*}
Then, $\alpha_j=\alpha^+_{2j+1}$ can be expressed in terms of $\{x_j\}$, see~\cite{ZZ-SIGMA-2017}.
The first few $f_{[N]}$ and $g_{[N]}$ are
\begin{gather*}
 f_{[1]}=x_1, \qquad g_{[1]}=1, \qquad
 f_{[2]}=\frac{x_1^3-x_3}{3},\qquad g_{[2]}=x_1^2, \\
 f_{[3]}=\frac{1}{45}x_1^6-\frac{1}{9}x_1^3x_3+\frac{1}{5}x_1x_5-\frac{1}{9}x_3^2,\qquad
g_{[3]}=\frac{2}{15}x_1^5-\frac{1}{3}x_1^2x_3+\frac{1}{5}x_5.
\end{gather*}
It has been proved in \cite{ZZ-SIGMA-2017} that $f_{[N]}$ and $g_{[N]}$
only depend on $\{x_1, x_3, \dots, x_{2N-1}\}$ and are homogeneous with degrees
 \[\mathcal{D}[f_N]=\frac{N(N+1)}{2},\qquad \mathcal{D}[g_N]=\frac{N(N+1)}{2}-1,\]
defined by the formula $\mathcal{D}\big[\prod\limits_{i\geq 1} x_i^{k_i}\big]=\sum\limits_{i\geq 1} ik_i$.
The function $u$ given by \eqref{H1-trans} is an odd function, and so it provides a solution to the nonlocal H1~\eqref{nlH1-2e}.

We remark that rational solutions in terms of $\{x_j\}$ have been obtained
for all the ABS equations except Q4 \cite{ZZ-SIGMA-2017,ZZ-JNMP-2019}.
This implies rational solutions for nonlocal ABS equations can be derived, which will be explored elsewhere.

\section{Conclusion} \label{sec-7}

We have presented a systematic way to generate multi-component lattice equations which are CAC, by making use of the cyclic group. We note that cyclic matrices have been used in 3-point differential-difference equations \cite{Babalic-JPA-2017,Carstea-JPA-2015} and fully discrete Lax pairs \cite{Fordy-JPA-2017} to generate multi-component systems.

Although the multi-component extensions we have considered are more general than ``the trivial Toeplitz extension'' introduced in Appendix~B of the PhD thesis of J.~Atkinson~\cite{Jat}, the key idea is the same: starting from a single-component CAC (2D or 3D) discrete system \eqref{consist-1}, replacing $u$ by a diagonal matrix \eqref{u-form}, applying permutations $T_k$ on shifted field components, one obtains a multi-component CAC system \eqref{consist-2'}. Posing the equations on lattices, D4 symmetry, criteria for decoupled and non-decoupled cases, BTs, auto-BTs, Lax pairs, solutions, nonlocal reductions, elimination of components to get equations on larger stencils, and reduction to multi-component discrete Painlev\'e equations, have been investigated in detail.

Isolated examples of multi-component extensions sporadically appear in the literature in different contexts. We mention: the two-component potential KdV system \cite[Table~5]{BHQK12}, cf.~\cite{FZ14}; a~two-component extension \cite[equation~(3.17)]{JLN} of the (non-potential) lattice mKdV equation \cite[equation~(2.49)]{NAH}; and the linear system of tetrahedral equations \cite[equation~(30)]{King-Schief-2015} that constitutes a two-component version of the BKP equation. We have provided a basic understanding for all these examples.

What we have not touched upon, is the fact that multi-component extension can also be applied to systems of equations. For example, the discrete Boussinesq family contains several multi-component CAC lattice systems \cite{Hietarinta-JPA-DBSQ}. These also allow $N[a,b]$ extension, cf.\ the $2[1,1]$ Boussinesq extension given in~\cite{FZ14}.

Finally, we'd like to point out that the multi-component extensions considered here are commutative. Non-commutative lattice equations, which are multi-component generalisations, have also been considered in the literature: 3D matrix discrete interable systems can be found in \cite[equations~(1.5)--(1.9)]{NC}, a~non-commutative version of~$\text{Q1}^0$ was given in \cite[equation~(23)]{BS2}, a~matrix version of H1 was obtained in \cite[equation~(2.1)]{FNC}, cf.~\cite{DJZ}, where several matrix discrete integrable equations derived from the Cauchy matrix approach were presented.

\appendix

\section{Multi-component Lax pair from BT}\label{appendixA}
For the scalar consistent lattice system \eqref{consist-1} with $Q^*=Q$, there exist functions $G_i$ such that
\begin{gather}\label{u-G}
 \th u=G_1(u,\t u,\h u),\qquad
 \tb u=G_2(u,\t u,\b u),\qquad
\hb u=G_3(u,\b u,\h u),
\end{gather}
which leads to the Lax pair \eqref{Lax-LM1} for \eqref{Q0}.
It holds as well if $u$ is extended to the diagonal form \eqref{u-form}.
Introduce $\b u=gf^{-1}$ where $f$ and $g$ are given in~\eqref{fg}.
From \eqref{u-G} there exist functions $F_2$ and $F_3$ such that
\[ \t g\t f^{-1}=F_2(u,\t u, f, g),\qquad \h g\h f^{-1}=F_3(u,\h u, f, g),\]
which leads to a Lax pair for \eqref{Q} where $u$ is~\eqref{u-form}:
\begin{gather}\label{Lax-LM3}
\Phi\big(\t f,\t g\big)=L(u, \t u)\Phi(f,g),\qquad \Phi\big(\h f,\h g\big)=M(u, \h u)\Phi(f,g),
\end{gather}
where
\begin{gather}\label{Phi}
\Phi(f,g)=(f_1,g_1,f_2,g_2,\dots,f_N,g_N)^{\rm T},
\end{gather}
and $L$ and $M$ are defined by \eqref{LM} through $\mathcal{L}$ and $\mathcal{M}$.

Meanwhile, from \eqref{u-G} we have
\begin{gather*}
 T_{a+b}\th u =G_1(u,T_{a}\t u,T_{b}\h u),\qquad
 T_{a+c}\tb u =G_2(u,T_{a}\t u,T_c \b u),\qquad
 T_{b+c}\hb u =G_3(u,T_{c}\b u,T_{b}\h u),
\end{gather*}
and consequently
\begin{gather*}
T_{a+c}\big(\t g\t f^{-1}\big)=F_2(u,T_{a}\t u, T_{c}f, T_{c}g),\qquad
T_{b+c}\big(\h g\h f^{-1}\big)=F_3(u,T_{b}\h u, T_{c}f, T_{c}g).
\end{gather*}
From this and the definition \eqref{Phi} for $\Phi$, we obtain (with $\theta$ defined in Theorem~\ref{T4})
\begin{gather*}
\Phi( T_{c}f, T_{c}g)=\theta^c \Phi(f, g),\qquad
\Phi\big( T_{a+c}\t f, T_{a+c}\t g \big)=\theta^{a+c} \Phi\big(\t f,\t g\big),\\
\Phi\big( T_{b+c}\h f, T_{b+c}\h g \big)=\theta^{b+c} \Phi\big(\h f,\h g\big).
\end{gather*}
Compared with \eqref{Lax-LM3}, it comes to
\[\theta^{a+c} \Phi\big(\t f,\t g\big)=L(u, T_{a}\t u)\theta^c \Phi(f, g),\qquad \theta^{b+c} \Phi\big(\h f,\h g\big)=M(u, T_{b}\h u)\theta^c \Phi(f, g),\]
which gives rise to the Lax pair \eqref{Lax-LM2} for~\eqref{QQ}.

\section{Proof of Theorem \ref{T2}} \label{adec}
We first prove a useful lemma. Although we believe it is an elementary result in number theory, we include it
for completeness. We denote $\mathbb{N}=\{1,2,\ldots\}$ and $\mathbb{N}^0=\{0,1,2,\ldots\}$.
\begin{Lemma}\label{L2}
Let $a,b\in\mathbb{N}^0$, $N \in \mathbb{N}$ such that $\operatorname{gcd}(a,b,N)=1$, and let $\mathbb{A}=\{ i a+ j b + k N\colon i,j,k\in\mathbb{Z}\}$.
There are $i_0,j_0\in \mathbb{N}$ and $k_0\in\mathbb{Z}$
such that
\begin{gather} \label{ijk-1}
1=i_0 a+j_0 b +k_0 N,
\end{gather}
and hence $\mathbb{A}=\mathbb{Z}$.
\end{Lemma}

\begin{proof}
Let $s_0=i_0 a+j_0 b+k_0N$ be the smallest positive integer in $\mathbb{A}$. For all $s>0 \in \mathbb{A}$, there are $i, j, k \in \mathbb{Z}$
such that $s=i a+j b+k N$, and there exist $q, r\in \mathbb{Z}$ such that $s=s_0 q+r$ where $q>0$ and $0\leq r < s_0$.
Then we have
\begin{gather*}
0\leq r = s-s_0 q = (i-i_0 q) a +(j-j_0 q) b+(k-k_0 q) N \in \mathbb{A}.
\end{gather*}
Since $s_0$ is the smallest positive number in $\mathbb{A}$ and $0\leq r < s_0$,
we must have $r=0$, which leads to $s=s_0q$. As $s$ was arbitrary it follows that $s_0$ is a divisor of $\operatorname{gcd}(a,b,N)$,
which implies $s_0=1$ because $\operatorname{gcd}(a,b,N)=1$. Thus we reach~\eqref{ijk-1} and consequently $\mathbb{A}$ covers $\mathbb{Z}$.
If $i_0$ and $j_0$ are not positive, there exist $i, j\in \mathbb{N}$ such that
$i N + i_0>0$ and $j N+j_0>0$, and from \eqref{ijk-1} we have
\[1=(i N+i_0)a +(j N+j_0)b + (k_0-i a -j b)N.\tag*{\qed}
\]\renewcommand{\qed}{}
\end{proof}

We now prove Theorem~\ref{T2}.

\begin{proof} We consider two cases, $d=\operatorname{gcd}(a,b,N)>1$ and $d=1$.

$d>1$: The $N[a,b]$ system \eqref{QQ}, viewed as a set which we denote by $S$ here, contains $N$ equations of the form
\begin{gather*}%\label{QQQ}
Q[k]\colon \ Q\big(u_{k},\t u_{k+a},\h u_{k+b},\th u_{k+a+b}\big)=0,\qquad k=1,2,\dots, N,
\end{gather*}
where the sub-index $i$ on the components $u_i$ and hence the `argument' of $Q[\cdot]$) is taken modulo~$N$. These equations depend on $N$ field components,
\[
V=\{u_i\colon 1\leq i\leq N\}.
\]

For each $1\leq i\leq d$ we define a subset of $M=N/d$ equations
\[
S_i=\{Q[i],Q[d+i],\ldots,Q[(M-1)d+i]\},
\]
so that the disjoint union $\cup_{i=1}^d S_i$ equals the set $S$. Writing the set of variables as a disjoint union, $V=\cup_{i=1}^d V_i$ where
$u_j\in V_i$ iff $j\equiv i$ mod $N$, we have, for all $i$, that each equation in $S_i$ only depends on the variables in $V_i$. Renaming
the variables $u_{i+dj}\in V_i$ by $v_j$ shows that the system $S_i$ is a $M[a/d,b/d]$ system.

 $d=1$: We distinguish three cases: $a=0$, $a=b$, $a<b$.

 $a=0$: The generic equation in the system $N[0,b]$, with $b\neq0$, is
\[
Q[k]\colon \ Q\big(u_{k},\t u_{k},\h u_{k+b},\th u_{k+b}\big)=0.
\]
Suppose $Y$ is a subset of equations which depend on a subset of variables $U\subset V$. If all equations that depend on variables in $U$ are in $Y$ and $Y$ is a proper subset, then the system is decoupled. Without loss of generality, suppose $Q[1]\in Y$. As $Q[1]$ depends on $u_{1+b}$, we have $Q[1+b]\in Y$, which in turn implies that $Q[1+2b]\in Y$. Continuing this argument
\begin{equation} \label{Ts}
Y\supset \{Q[1+ib]\colon 0\leq i < N\}
\end{equation}
contains $N$ equations. As $1+ib\equiv 1+jb$ mod $N$ implies $i\equiv j$ mod $N$ when $\operatorname{gcd}(b,N)=1$, they are all distinct. This shows that $Y$ is not a proper subset and hence the system is non-decoupled.

$a=b$: The generic equation in the system $N[b,b]$, with $b\neq0$, is
\[
Q[k]\colon \ Q\big(u_{k},\t u_{k+b},\h u_{k+b},\th u_{k+2b}\big)=0.
\]
As in the case $a=0$, attempting to construct a proper subset of equations, $Y$, we find~\eqref{Ts}.

$a<b$: For the system $N[a,b]$, with $0<a<b$, the generic equation has the form
\[
Q[k]\colon \ Q\big(u_{k},\t u_{k+a},\h u_{k+b},\th u_{k+a+b}\big)=0.
\]
Starting from $Q[1]\in Y$, following the dependence on the variables we must have
\[
\big\{Q[1+ai+bj]\colon i,j\in \mathbb{N}^0\big\} \subset Y.
\]
It follows from Lemma~\ref{L2} that $Y$ contains $N$ distinct equations and therefore this case is non-decoupled as well.
\end{proof}

\subsection*{Acknowledgments}
The authors thank Jarmo Hietarinta for his suggestion to include equations (\ref{JH1}), (\ref{JH2}) and noting they are not of the same form. We thank Pavlos Kassotakis and Maciej Nieszporski for noting equation (\ref{201H1}) can be written as a quadrilateral system. We thank all referees for their comments, especially the referee who pointed out Appendix~B from~\cite{Jat}. DJZ is grateful to Professors Q.P.~Liu and R.G.~Zhou for warm discussion. This project is supported by the NSF of China (grant nos.~11875040, 11631007 and 11801289), the K.C.~Wong Magna Fund in Ningbo University, and a CRSC grant from La Trobe University.

%\cite{JH-JNMP-cac,Hirota-JPSJ-1973,Hir-Sats-JPSJ-1976,Jim-Sak-LMP-1996,Wadati-PTPS-1976}

\pdfbookmark[1]{References}{ref}
\LastPageEnding

\end{document}